%% file: main.tex
\documentclass[12pt]{article}
\usepackage[margin=1in]{geometry}
\usepackage{setspace}
\usepackage[margin=1.5cm, labelfont=bf]{caption}
\usepackage{subcaption}
\usepackage{graphicx}
\usepackage{fancyhdr}
\usepackage{xcolor}

\usepackage{titlesec}
\titleformat*{\section}{\Large\scshape}
\titleformat*{\subsection}{\large\scshape}
\titleformat*{\subsubsection}{\normalsize\itshape}

\usepackage{amssymb,amsmath,amsthm,amsfonts,mathtools,array}
\usepackage{verbatim, tikz, bbm}
\usepackage{xfrac} 
\usepackage[mathscr]{euscript} 
\usepackage{empheq} 
\usepackage{enumitem}
\usepackage{subeqnarray}

\theoremstyle{plain} \numberwithin{equation}{section}
\newtheorem{theorem}{Theorem}[section]

\newtheorem{conjecture}[theorem]{Conjecture}
\newtheorem{lemma}[theorem]{Lemma}

\usepackage{url}
\usepackage{hyperref}
\definecolor{linkblue}{HTML}{1EB5EB}
\definecolor{linkgreen}{HTML}{06BA63}
\definecolor{linkred}{HTML}{E3170A}

\hypersetup{
	pdftitle={Long-Wavelength Limit of Two-Fluid EP System},
	pdfauthor={Emily Kelting},
	pdfsubject={partial differential equations},
	pdfkeywords={Plasma, Euler-Poisson, Ion-Acoustic Waves, Korteweg-de Vries},
	colorlinks=true,
	linkcolor=linkblue,
	citecolor=linkgreen,
	urlcolor=linkred
}
\usepackage{authblk}
\usepackage[nameinlink]{cleveref}
\usepackage{bigints}
\usepackage{rotating}

\catcode`,\active

\catcode`\,12

\usepackage{twemojis}

\allowdisplaybreaks

\begin{document} 
\thispagestyle{empty}
\setstretch{1.15}

\begin{center}
    \textsc{\Large{Long-Wavelength Limit of the Two-Fluid Euler-Poisson System}}
\\[20pt]
    {\large{Emily K. Kelting\textsuperscript{\small{\twemoji{cloud with lightning}}} and J. Douglas Wright\textsuperscript{\footnotesize{\twemoji{dragon}}}}}
    \\[10pt]
\end{center}
\begin{minipage}[c]{0.45\textwidth}
\centering\textsuperscript{\small{\twemoji{cloud with lightning}}}\textit{Department of Mathematics \\ University of New England \\ Biddeford, ME 04005}
\end{minipage}
\hfill
\begin{minipage}[c]{0.45\textwidth}
\centering\textsuperscript{\footnotesize{\twemoji{dragon}}}\textit{Department of Mathematics \\ Drexel University \\ Philadelphia, PA 19104}
\end{minipage}
\\[10pt]

\hrule
\begin{abstract}
    Plasma is a medium filled with free electrons and positive ions. Each particle acts as a conducting fluid with a single velocity and temperature when electromagnetic fields are present. This distinction between the roles played by electrons and ions is what we refer to as the two-fluid description of plasma. In this paper, we investigate the dynamics of these particles in both hot and cold plasma using a collisionless ``Euler-Poisson'' system. Employing analytical and computational techniques from differential equations, we show this system is governed by the dynamics of the Korteweg–de Vries (KdV) equation in the long-wavelength limit.
\end{abstract}

\\[6pt]
\hrule

\pagestyle{fancy}
\input{intro.tex}
\input{kdv.tex}
\input{conclusion}

\newpage
\bibliographystyle{siam}

\input{main.bbl}
\newpage
\appendix{
\section*{Appendix}
\addcontentsline{toc}{section}{Appendix}
\input{appendix.tex}
}

\end{document}

%% file: intro.tex
\section{Introduction}

Plasma is a complex medium consisting of freely flowing charged elementary particles -- cations and electrons -- in a neutral background \cite{waymouth1991}. With electric fields present, each group behaves like a conducting fluid with its own temperature. This distinction between the roles played by ions and electrons is referred to as the two-fluid description of plasma \cite{krishan,reitz1960}. 

In our model, we ignore gravitational and magnetic fields, and let the particles be affected only by the electrical field as they move across space. We will assume, moreover, that the electric field causes the particles to diverge --- a collisionless system. This simplified representation of plasma, incorporating the Euler equations to describe fluid motion and the Poisson equation to characterize electric potential, is aptly named the collisionless ``Euler-Poisson'' (EP) system \cite{krishan,reitz1960}.

\subsection{Two-fluid Euler-Poisson System}\label{EP_2}

Letting \(n_{\pm}(x,t)\) be the particle number density at spatial point \(x\in\mathbb{R}\) and time \(t\geq 0\), \(v_{\pm}(x,t)\) be the velocity, and \(\phi(x,t)\) be the electric potential, with ``+'' representing the cations and ``--'' the electrons, the EP system is
\begin{equation}
    \label{eq::5hotrescale}
    \begin{aligned}
        \partial_t n_+ + \partial_x(n_+ v_+) = 0, 
         \\[4pt]
        \partial_t n_- + \partial_x(n_- v_-) = 0,  \\[4pt]
        \partial_t v_+ + \tfrac{1}{2}\partial_x v_+^2 + \tau_i \partial_x \ln(n_+) + \partial_x\phi = 0,
        \\[4pt]
        m_e \left(\partial_t v_- + \tfrac{1}{2}\partial_x v_-^2\right) + \partial_x \ln(n_-) - \partial_x\phi = 0,
         \\[4pt]
        \partial_{xx}\phi - n_- + n_+ = 0.
    \end{aligned}
\end{equation}
\noindent Here, the parameter \(m_e\) denotes the mass of the electron, and \(\tau_i\) is the temperature of the ion. It is important to note that the model is scaled so that the ion mass is at unity. Hence, the corresponding electron mass is \(m_e \approx 0.00054551\). Furthermore, by scaling the electron temperature to unity, we can identify the plasma as ``hot'' when \(\tau_i=1\)  (thermodynamic equilibrium) and ``cold'' when \(\tau_i=0\) \cite{krishan}.

\subsubsection{Galilean Ivariance \& Reversibility}\label{EP_2_Gali}
A noteworthy observation about the equations in \eqref{eq::5hotrescale} is that they are reversible and Galilean invariant. That is, if \((n_{\pm}, v_{\pm}, \phi)(x, t)\) is a solution, then so are \((n_{\pm}, v_{\pm}, \phi)(-x, -t)\) and \((n_{\pm}, v_{\pm}+c, \phi)(x-ct, t)\), respectively. The reversibility of the equations holds a key place in our research as it captures the symmetrical nature of the physical laws and processes of plasma. It means \eqref{eq::5hotrescale} can accurately predict the EP system's behavior irrespective of the direction of time. Moreover, Galilean invariance is a crucial principle, asserting that the motion of the equations of \eqref{eq::5hotrescale} are consistent and applicable across various reference frames.

\subsection{Motivation}

A typical model-simplifying assumption is to set \(m_e = 0\) in \eqref{eq::5hotrescale} as the mass of an electron is many orders of magnitude smaller than a positive ion. With this estimation, the electron distribution behaves according to the Boltzmann Law, \(n_- = n_0\text{e}^{\phi}\) where \(n_0\) is the equilibrium particle density. Implementing these estimations, the electron terms disappear and the plasma is classified as a singular fluid. 

Experimental and theoretical studies have shown that the long-wavelength solutions of both hot and cold states of the single-fluid Euler-Poisson system converge globally in time to the solutions of the Korteweg-de Vries equation. However, in the twentieth century, only formal derivations of the long-wavelength limit were known \cite{washimi1966, su1969}. Guo and Pu justified this rigorously in 2014 using \textit{a priori} estimates and energy arguments \cite{guo}. 
Nonetheless, since electrons have mass, the question is: do these findings still hold if we take \(m_e \neq 0\)?

%% file: kdv.tex
\section{The Korteweg-de Vries Approximation}\label{kdv}

In plasma dynamics, the concept of the long-wavelength limit is crucial as it allows us to understand the behavior of plasma when perturbed by low-frequency waves. The term ``long-wavelength limit'' describes a situation in which the wave's wavelength is significantly greater than the scale of the system in a spatially uniform electric field. 
In such circumstances, numerous model simplifications can be applied, leading to a more tractable analysis.




\subsection{Transformation \& Formal Expansion}\label{kdv_2}

Just as in \cite{guo}, we begin by applying the Gardner-Morikawa Transformation to \eqref{eq::5hotrescale}, 
\begin{equation*} 
Y = \varepsilon^{\sfrac{1}{2}}(x-ct), \qquad T = \varepsilon^{\sfrac{3}{2}}t,
\end{equation*}
and acquiring the parameterized system
\begin{equation}\label{eq::5param}
    \begin{split}
        \varepsilon\partial_{T}n_{+} - c\partial_{Y}n_{+} + \partial_{Y}(n_{+}v_{+}) & = 0, \\[3pt]
        \varepsilon\partial_{T}n_{-} - c\partial_{Y}n_{-} + \partial_{Y}(n_{-}v_{-}) & = 0, \\[3pt]
        \varepsilon\partial_{T}v_{+} - c\partial_{Y}v_{+} + v_{+}\partial_{Y}v_{+} + \tau_i \partial_{Y}\ln(n_{+}) + \partial_{Y}\phi & = 0, \\[3pt]
        m_e(\varepsilon\partial_{T}v_{-} - c\partial_{Y}v_{-} + v_{-}\partial_{Y}v_{-}) + \partial_{Y}\ln(n_{-}) - \partial_{Y}\phi & = 0, \\[3pt]
        \varepsilon \partial_{Y}^{2}\phi + n_{+}-n_{-} & = 0 
    \end{split}
\end{equation}%
\cite{su1969}. Here, \(\varepsilon\) is the amplitude of the initial disturbance, with \(0 < \varepsilon \ll 1\), and \(c\) is a velocity parameter to be determined later. Consider the following formal expansion centered around the equilibrium solution \((n_{\pm},v_{\pm},\phi)=(1,0,0)\),
\begin{subequations}\label{eq::form_exp}
\begin{align}
    n_{\pm} & = 1 + \varepsilon^1 n_{\pm}^{(1)} + \varepsilon^2 n_{\pm}^{(2)} + \varepsilon^3 n_{\pm}^{(3)} + \cdots, \\[3pt]
    v_{\pm} & = \varepsilon^1 v_{\pm}^{(1)} + \varepsilon^2 v_{\pm}^{(2)} + \varepsilon^3 v_{\pm}^{(3)} + \cdots, \\[3pt]
    \phi & = \varepsilon^1 \phi^{(1)} + \varepsilon^2 \phi^{(2)} + \varepsilon^3 \phi^{(3)} + \cdots .
\end{align}
\end{subequations}
Using this replacement for the variables in \eqref{eq::5param}, we get a power series of \(\varepsilon\) whose coefficients depend on \((n_{\pm}^{(k)},v_{\pm}^{(k)},\phi^{(k)})\) for \(k\geq 0\). We now match the coefficients to determine the functions \(n_\pm^{(k)}\), \(v_\pm^{(k)}\), \(\phi^{(k)}\).


\subsubsection{The coefficients of \(\varepsilon^0\)}
The coefficients of \(\varepsilon^0\) are automatically 0.


\subsubsection{The coefficients of \(\varepsilon^1\)} 
Setting the coefficients of \(\varepsilon^1\) to be 0, we get the system
\begin{subequations}\label{eq::coeff_e1}
\begin{empheq}[left = (\mathscr{E}_0) \empheqlbrace\,]{align}
        \partial_{Y}v_{+}^{(1)} - c\partial_{Y}n_{+}^{(1)} & = 0, \\[5pt]
        \partial_{Y}v_{-}^{(1)} - c\partial_{Y}n_{-}^{(1)} & = 0, \\[5pt]
        \tau_i \partial_{Y}n_{+}^{(1)} -c\partial_{Y}v_{+}^{(1)} + \partial_{Y}\phi^{(1)} & = 0, \\[5pt]
        \partial_{Y}n_{-}^{(1)} - cm_e\partial_{Y}v_{-}^{(1)} - \partial_{Y}\phi^{(1)} & = 0, \\[5pt]
        n_{+}^{(1)} - n_{-}^{(1)} & = 0.
\end{empheq}
\end{subequations}

This set of equations written as a matrix is
\begin{equation*}
\begin{bmatrix}
    -c & 0 & 1 & 0 & 0 \\
    0 & -c & 0 & 1 & 0 \\
    \tau_i & 0 & -c & 0 & 1 \\
    0 & 1 & 0 & -cm_e & -1 \\
    1 & -1 & 0 & 0 & 0
\end{bmatrix}
\begin{bmatrix}
    \partial_{Y}n_{+}^{(1)} \\
    \partial_{Y}n_{-}^{(1)} \\
    \partial_{Y}v_{+}^{(1)} \\
    \partial_{Y}v_{-}^{(1)} \\
    \partial_{Y}\phi^{(1)} \\
\end{bmatrix}
\;=\;
0.
\end{equation*}
To get a nontrivial solution for \(n_{\pm}^{(1)}, \, v_{\pm}^{(1)}, \text{ and } \phi^{(1)}\), we need the determinant of the coefficient matrix to be zero, which implies 
\begin{equation}\label{eq::C2}
    c^2 = \frac{1+\tau_i}{1+m_e}.
\end{equation} 
We define the root of \eqref{eq::C2} as the `speed of sound', the maximum propagation speed of waves associated with linear dispersion. Because of the Galilean invariance, we are allowed to take the positive root for simplicity. This will be the value of \(c\) dedicated to future calculations. 

Better still, \eqref{eq::coeff_e1} gives the relation
\begin{empheq}[left = (\mathscr{S}_1) \empheqlbrace\,]{equation}
\label{eq::rely_on_n+1}
    \begin{split}
        n_{-}^{(1)} & = n_{+}^{(1)}, \\
        v_{+}^{(1)} & = cn_{+}^{(1)}, \\
        v_{-}^{(1)} & = cn_{+}^{(1)}, \\
        \phi^{(1)} & = (1 - \tau_i m_e) n_{+}^{(1)}. 
    \end{split}
\end{empheq}
Therefore, to determine \((n_{\pm}^{(1)},v_{\pm}^{(1)},\phi^{(1)})\), we need only establish \(n_{+}^{(1)}\).


\subsubsection{The coefficients of \(\varepsilon^2\) and the KdV equation for \(n_{+}^{(1)}\)}
Setting the coefficients of \(\varepsilon^2\) to be zero, we obtain the evolution system
\begin{subequations}\label{eq::coeff_e2}
\begin{empheq}[left = (\mathscr{E}_1) \empheqlbrace\,]{align}
    \label{eq::coeff_e2_a}
        \partial_{T}n_{+}^{(1)} + \partial_{Y}v_{+}^{(2)} + \partial_{Y}\left(n_{+}^{(1)}v_{+}^{(1)}\right) - c\partial_{Y}n_{+}^{(2)} = 0, \\[3pt]
    \label{eq::coeff_e2_b}
        \partial_{T}n_{-}^{(1)} + \partial_{Y}v_{-}^{(2)} + \partial_{Y}\left(n_{-}^{(1)}v_{-}^{(1)}\right) - c\partial_{Y}n_{-}^{(2)} = 0, \\[3pt]
    \label{eq::coeff_e2_c}
        \tau_i \partial_{Y}\left(n_{+}^{(2)} - \frac{1}{2}(n_{+}^{(1)})^2\right) + \partial_{T}v_{+}^{(1)} + v_{+}^{(1)}\partial_{Y}v_{+}^{(1)} - c\partial_{Y}v_{+}^{(2)} + \partial_{Y}\phi^{(2)} = 0, \\[3pt]
    \label{eq::coeff_e2_d}
        \partial_{Y}\left(n_{-}^{(2)} - \frac{1}{2}(n_{-}^{(1)})^2\right) + m_e\left(\partial_{T}v_{-}^{(1)} + v_{-}^{(1)}\partial_{Y}v_{-}^{(1)} - c\partial_{Y}v_{+}^{(2)}\right) - \partial_{Y}\phi^{(2)} = 0, \\[3pt]
    \label{eq::coeff_e2_e}
        \partial_{Y}^{2}\phi^{(1)} + n_{+}^{(2)} - n_{-}^{(2)} = 0.
\end{empheq}
\end{subequations}
%
%
Differentiating \eqref{eq::coeff_e2_e} with respect to \(Y\) and multiplying it by \(c^2(1-\tau_i m_e)\), 
as well as multiplying \eqref{eq::coeff_e2_a} by \(c(1+\tau_i)\) and \eqref{eq::coeff_e2_b} by \(c m_e(1+\tau_i)\), and then adding all of these together with \((1+\tau_i)\)\eqref{eq::coeff_e2_c} and \((1+\tau_i)\)\eqref{eq::coeff_e2_d}, we find \(n_{+}^{(1)}\) satisfies the classic Korteweg-de Vries equation,
\begin{equation}\label{eq::KdV_n+1}
    \partial_{T}n_{+}^{(1)} + c n_{+}^{(1)}\partial_{Y}n_{+}^{(1)} + \frac{(c^2 m_e -1)^2}{2c(1+m_e)} \partial_{Y}^{3}n_{+}^{(1)} = 0.
\end{equation}
Utilization of relations \eqref{eq::C2} and \eqref{eq::rely_on_n+1} have caused all of the coefficients of \(n_{\pm}^{(2)}, \, v_{\pm}^{(2)}, \text{ and } \phi^{(2)}\) in \eqref{eq::KdV_n+1} to vanish. 

Observe that systems \eqref{eq::rely_on_n+1} and \eqref{eq::KdV_n+1} are self-contained and do not depend on \((n_{\pm}^{(m)},v_{\pm}^{(m)},\phi^{(m)})\) for \(m\geq 2\). Thus, taking \((n_{\pm}^{(1)}, v_{\pm}^{(1)}, \phi^{(1)})\) as known solutions solved by \eqref{eq::rely_on_n+1} and \eqref{eq::KdV_n+1}, we can find the equations satisfied by \((n_{\pm}^{(2)}, v_{\pm}^{(2)}, \phi^{(2)})\) from \eqref{eq::coeff_e2},
\begin{empheq}[left = (\mathscr{S}_2) \empheqlbrace\,]{equation}
\label{eq::rely_on_n+2}
    \begin{array}{*2{>{\displaystyle}l}}
        n_{-}^{(2)} & = \displaystyle{n_{+}^{(2)} + (1 - \tau_i m_e) \partial_{Y}^2 n_{+}^{(1)}},\\[10pt]
        v_{+}^{(2)} & = cn_{+}^{(2)} + \frac{(c^2 m_e-1)^2}{2c(1+m_e)}\partial_{Y}^2 n_{+}^{(1)} - \frac{c}{2}(n_{+}^{(1)})^2 , \\[10pt]
        v_{-}^{(2)} & = cn_{+}^{(2)} +\frac{c(1-\tau_i m_e)}{1+m_e}\partial_{Y}^2 n_{+}^{(1)} - \frac{c}{2}(n_{+}^{(1)})^2,\\[10pt]
        \phi^{(2)} & = (1 - \tau_i m_e)n_{+}^{(2)} +\frac{(c^2 m_e-1)^2}{1+m_e}\partial_{Y}^2 n_{+}^{(1)} - \frac{1 - \tau_i m_e}{2}(n_{+}^{(1)})^2.
    \end{array}
\end{empheq}
Notice that \eqref{eq::rely_on_n+1} and \eqref{eq::rely_on_n+2} have the same structure, with the latter containing lower-order correcting terms dependent only on \(n_+^{(1)}\). Again, to determine \((n_{\pm}^{(2)},v_{\pm}^{(2)},\phi^{(2)})\), we need only find \(n_{+}^{(2)}\) and \((n_{\pm}^{(1)},v_{\pm}^{(1)},\phi^{(1)})\).


\subsubsection{The coefficients of \(\varepsilon^3\) and the linearized KdV equation for \(n_+^{(2)}\)}
\noindent Setting the coefficients of \(\varepsilon^3\) to be zero returns the evolution system
\begin{subequations}\label{eq::coeff_e3}
\begin{empheq}[left = (\mathscr{E}_2) \empheqlbrace\,]{align}
    \label{eq::coeff_e3_a}
    \partial_{T} n_{+}^{(2)} + \partial_{Y} v_{+}^{(3)} + \partial_{Y}\left( n_{+}^{(1)}v_{+}^{(2)} + n_{+}^{(2)}v_{+}^{(1)}\right) - c \partial_{Y} n_{+}^{(3)} = 0, \\[3pt]
    \label{eq::coeff_e3_b}
    \partial_{T} n_{-}^{(2)} + \partial_{Y} v_{-}^{(3)} + \partial_{Y}\left( n_{-}^{(1)}v_{-}^{(2)} + n_{-}^{(2)}v_{-}^{(1)} \right) - c \partial_{Y} n_{-}^{(3)} = 0, \\[3pt]
    \nonumber
    \partial_{Y}\left(n_{+}^{(3)} - n_{+}^{(1)}n_{+}^{(2)} + \frac{1}{3}\left(n_{+}^{(1)}\right)^3 \right) \qquad \\
    \label{eq::coeff_e3_c}
    \partial_{T} v_{+}^{(2)} + \partial_{Y}\left( v_{+}^{(1)}v_{+}^{(2)} - c v_{+}^{(3)}\right) + \partial_{Y} \phi^{(3)} = 0, \\[3pt]
    \nonumber
    \partial_{Y}\left(n_{-}^{(3)} - n_{-}^{(1)}n_{-}^{(2)} + \frac{1}{3}\left(n_{-}^{(1)}\right)^3 \right) \qquad \\ 
    \label{eq::coeff_e3_d}
     + m_e\left(\partial_{T} v_{-}^{(2)} + \partial_{Y}\left(v_{-}^{(1)}v_{-}^{(2)} - c v_{-}^{(3)}\right)\right) - \partial_{Y} \phi^{(3)} = 0, \\[3pt]
    \label{eq::coeff_e3_e}
    \partial_{Y}^2 \phi^{(2)} + n_{+}^{(3)} - n_{-}^{(3)} = 0.
\end{empheq}
\end{subequations}
Differentiating \eqref{eq::coeff_e3_e} with respect to \(Y\) and multiplying it by \(c^2(1-\tau_i m_e)\), as well as multiplying \eqref{eq::coeff_e3_a} by \(c(1+\tau_i)\) and \eqref{eq::coeff_e3_b} by \(c m_e(1+\tau_i)\), and then adding all of these together with \((1+\tau_i)\)\eqref{eq::coeff_e3_c} and \((1+\tau_i)\)\eqref{eq::coeff_e3_d}, we find \(n_{+}^{(2)}\) satisfies the inhomogeneous linearized Korteweg-de Vries equation
\begin{equation}\label{eq::KdV_n+2}
    \partial_{T} n_{+}^{(2)} + c\partial_{Y}\left(n_{+}^{(1)}n_{+}^{(2)}\right) + \frac{(c^2 m_e -1)^2}{2c(1+m_e)^3}\partial_{Y}^3 n_{+}^{(2)} = F^{(1)}
\end{equation}
where \(F^{(1)}\) depends only on \(n_{+}^{(1)}\). Here we have used relations \eqref{eq::rely_on_n+2}, \eqref{eq::rely_on_n+1}, and \eqref{eq::C2} under which all of the coefficients of \(n_{\pm}^{(3)}, \, v_{\pm}^{(3)}, \text{ and } \phi^{(3)}\) in \eqref{eq::KdV_n+2} disappear. Once again, note that systems \eqref{eq::rely_on_n+2} and \eqref{eq::KdV_n+2} for \((n_{\pm}^{(2)},v_{\pm}^{(2)},\phi^{(2)})\) are self-contained and do not depend on \((n_{\pm}^{(m)},v_{\pm}^{(m)},\phi^{(m)})\) for \(m\geq 3\). As such, from \eqref{eq::coeff_e3}, we can express \((n_{-}^{(3)},v_{\pm}^{(3)},\phi^{(3)})\) in terms of \(n_+^{(3)}\) and \((n_{\pm}^{(k)},v_{\pm}^{(k)},\phi^{(k)})\), \(k=1,2\):
\begin{empheq}[left = (\mathscr{S}_3) \empheqlbrace\,]{equation}
\label{eq::rely_on_n+3}
    \begin{array}{*2{>{\displaystyle}l}}
        n_{-}^{(3)} & = n_{+}^{(3)} + f^{(2)}, \\[3pt]
        v_{+}^{(3)} & = cn_{+}^{(3)} + g^{(2)}, \\[3pt]
        v_{-}^{(3)} & = cn_{+}^{(3)} + h^{(2)}, \\[3pt]
        \phi^{(3)} & = (1 - \tau_i m_e) n_{+}^{(3)} + j^{(2)}, 
    \end{array}
\end{empheq}
where \(f^{(2)}, \, g^{(2)}, \, h^{(2)}, \, j^{(2)}\) are functions that depend only on \(n_{+}^{(1)}\) and \(n_{+}^{(2)}\). 


\subsubsection{The coefficients of \(\varepsilon^{k+1}\) and the linearized KdV equation for \(n_+^{(k)}\)} \label{sec::eps_k+1_coeff}
Setting the coefficients of \(\varepsilon^{k+1}\) to be zero, we get an evolution system for \((n_{\pm}^{(k)},v_{\pm}^{(k)},\phi^{(k)})\),
%
{\small{
\begin{subequations}\label{eq::coeff_ek}
\begin{empheq}[left = (\mathscr{E}_k) \empheqlbrace\,]{align}
    \label{eq::coeff_ek_a}
        \partial_{T}n_{+}^{(k)} + \partial_{Y}v_{+}^{(k+1)} + \partial_{Y}\left(\sum_{j=1}^{k} n_{+}^{(j)} v_{+}^{(k+1-j)}\right) - c\partial_{Y}n_{+}^{(k+1)} = 0, \\[3pt]
    \label{eq::coeff_ek_b}
        \partial_{T}n_{-}^{(k)} + \partial_{Y}v_{-}^{(k+1)} + \partial_{Y}\left(\sum_{j=1}^{k} n_{-}^{(j)} v_{-}^{(k+1-j)}\right) - c\partial_{Y}n_{-}^{(k+1)} = 0, \\[3pt]
    \label{eq::coeff_ek_c}
        \partial_{T} v_{+}^{(k)} - c\partial_{Y} v_{+}^{(k+1)} + \frac{1}{2}\partial_{Y}\left(\sum_{j=1}^{k} v_{+}^{(j)} v_{+}^{(k+1-j)}\right) + \tau_i \partial_Y \ell_{+}^{(k+1)} + \partial_{Y}\phi^{(k+1)} = 0, \\[3pt]
    \label{eq::coeff_ek_d}
        m_e \left(\partial_{T} v_{-}^{(k)} - c\partial_{Y} v_{-}^{(k+1)} + \frac{1}{2}\partial_{Y}\left(\sum_{j=1}^{k} v_{-}^{(j)} v_{-}^{(k+1-j)}\right) \right) + \partial_Y \ell_{-}^{(k+1)} - \partial_{Y}\phi^{(k+1)} = 0,\\[3pt]
    \label{eq::coeff_ek_e}
        \partial_{Y}^{2}\phi^{(k)} + n_{+}^{(k+1)} - n_{-}^{(k+1)} = 0,
\end{empheq}
\end{subequations}
}}
with 
\begin{equation}\label{eq::ell_def}
\begin{split}
    \ell^{(k+1)}_{\pm} = \sum_{j=1}^{k+1}\frac{(-1)^{j+1}}{j}\left(\sum_{\substack{a_1 + 2a_2 + \cdots + (k+1) a_{k+1} = k+1 \\ a_1 + a_2 + \cdots + a_{k+1} = j \\ a_1,a_2,\ldots,a_{k+1} \geq 0}} \binom{j}{a_1,a_2,\ldots,a_{k+1}} \prod_{i=1}^{k+1} \left(n_\pm^{(i)}\right)^{a_i}\right)
\end{split}
\end{equation}
%
%
%
where \[\binom{j}{a_1,a_2,\ldots,a_{k+1}} = \frac{j!}{a_1!a_2!\cdots a_{k+1}!}\] with \(a_i\in\mathbb{Z}\) for \(1\leq i \leq k+1\). The derivation of \(\ell_\pm^{(k+1)}\) utilizes a combination of Taylor's Theorem and Multinomial Theorem, an extension of the Binomial Theorem for powers greater than 2. For specifics, please see Appendix~\ref{app::ell_deriv}. From the evolution system, we obtain
\begin{empheq}[left = (\mathscr{S}_k) \empheqlbrace\,]{align}
\label{eq::5_rely_on_n+k}
    \begin{array}{*2{>{\displaystyle}l}}
        n_{-}^{(k+1)} & = n_{+}^{(k+1)} + f^{(k)} \\[3pt]
        v_{+}^{(k+1)} & = cn_{+}^{(k+1)} + g^{(k)},\\[3pt]
        v_{-}^{(k+1)} & = cn_{+}^{(k+1)} + h^{(k)},\\[3pt]
        \phi^{(k+1)} & = c v_{+}^{(k+1)} - \tau_i \ell_{+}^{(k+1)} + j^{(k)},
    \end{array}
\end{empheq}
for some functions \(f^{(k)}, \, g^{(k)}, \, h^{(k)}, \, j^{(k)}\) which only depend on \(n_{+}^{(m)}\) for \(1 \leq m \leq k\). The general forms for these functions are

\begin{empheq}[left = \empheqlbrace\,]{equation}
\label{eq::5_rely_on_n+k_spec}
    \begin{array}{*2{>{\displaystyle}l}}
        f^{(k)} & = \partial_{Y}^{2}\phi^{(k)} \\
        g^{(k)} & = - \int \partial_{T}n_{+}^{(k)} dY - \sum_{j=1}^{k} n_{+}^{(j)} v_{+}^{(k+1-j)},\\
        h^{(k)} & = - \int \partial_{T}n_{-}^{(k)} dY - \sum_{j=1}^{k} n_{-}^{(j)} v_{-}^{(k+1-j)},\\
        j^{(k)} & = -\int \partial_{T} v_{+}^{(k)} dY - \frac{1}{2}\sum_{j=1}^{k} v_{+}^{(j)} v_{+}^{(k+1-j)},\\
    \end{array}
\end{empheq}
The relations \eqref{eq::5_rely_on_n+k} and \eqref{eq::5_rely_on_n+k_spec} make the evolution system \eqref{eq::coeff_ek} solvable and show that we need only find \(n_{+}^{(k+1)}\) to get \((n_{\pm}^{(k+1)},v_{\pm}^{(k+1)},\phi^{(k+1)})\). 

By the same procedure that led to \eqref{eq::KdV_n+1} and \eqref{eq::KdV_n+2}, we arrive at the linearized inhomogeneous Korteweg-de Vries equation for \(n_+^{(k)}\),
\begin{equation}\label{eq::5_KdV_n+k}
    \partial_{T} n_{+}^{(k)} + c\partial_{Y}\left(n_{+}^{(1)}n_{+}^{(k)}\right) + \frac{(c^2 m_e -1)^2}{2c(1+m_e)}\partial_{Y}^3 n_{+}^{(k)} = F^{(k-1)}
\end{equation}
where \(F^{(k-1)}\) is a function that depends only on \(n_{+}^{(m)}\) for \(1 \leq m \leq k-1\), which are known from the previous \((k-1)\) steps. Once more, the system \eqref{eq::5_rely_on_n+k} and \eqref{eq::5_KdV_n+k} for \((n_{\pm}^{(k)},v_{\pm}^{(k)},\phi^{(k)})\) is self-contained and does not depend on \((n_{\pm}^{(j)},v_{\pm}^{(j)},\phi^{(j)})\) for \(j\geq k+1\).

\subsection{Solutions to KdV approximations}\label{kdv_2_sol}

Our previous calculations have resulted in multiple KdV approximations to the variables of our two-fluid plasma system \eqref{eq::5hotrescale}. All of the KdV equations, \eqref{eq::KdV_n+1}, \eqref{eq::KdV_n+2}, and \eqref{eq::5_KdV_n+k} are well-posed, and since the value of epsilon is small, we know that the formal expansion \eqref{eq::form_exp} is being dominated by \((n_{\pm}^{(1)},v_{\pm}^{(1)},\phi^{(1)})\). With the help of relationships from \eqref{eq::rely_on_n+1}, we can find our solution to \eqref{eq::KdV_n+1}, and obtain the governing piece. This fact is summarized in Theorem~\ref{thm::5_solve_n1KdV}.

\begin{center}{\rule{4cm}{0.4pt}}\end{center}
\begin{theorem}\label{thm::5_solve_n1KdV}
    Let \(s_1 \geq 2\) be a sufficiently large integer. Then for any given initial data \(n_{+}^{(1)}(Y,0) \in H^{s_1}(\mathbb{R})\) there exists a \(\tilde{T}>0\) such that the initial value problem \eqref{eq::KdV_n+1} with \eqref{eq::rely_on_n+1} has a unique solution 
    \begin{equation*}\left(n_{\pm}^{(1)}, v_{\pm}^{(1)}, \phi^{(1)} \right)\in L^\infty\left(H^{s_1}(\mathbb{R}); -\tilde{T}, \tilde{T}\right)
    \end{equation*} 
    with initial data \(\left(n_{\pm}^{(1)}(Y,0), cn_{\pm}^{(1)}(Y,0), (1 - \tau_i m_e)n_{+}^{(1)}(Y,0)\right)\). Furthermore, by the conservation laws of the KdV equation, we can extend the solution to any time interval \([-T,T]\).
\end{theorem}
\begin{proof}
    This is a classical result for the Korteweg-de Vries equation \cite{bona1975initial, gardner1967, hirota1981, ma2005}. In fact, such a solution to \eqref{eq::KdV_n+1} is 
    \begin{equation}\label{eq:kdv_n+1_sol}
        n_+^{(1)}(Y,T) \;=\; \frac{3\tilde\mu}{c} \text{sech}^2\left((Y - \tilde\mu T) \sqrt{\frac{c \tilde\mu \left(m_e+1\right)}{2\left(c^2 m_e-1\right)^2}}\right) 
    \end{equation}
    where \(\tilde\mu\) is the speed of the wave. The calculation is given in Appendix~\ref{app::kdv_sol}.
\end{proof}
\begin{center}{\rule{4cm}{0.4pt}}\end{center}

Furthermore, not only are the relationships in \eqref{eq::rely_on_n+1} subject to the solution of a Korteweg de-Vries equation, but the higher order terms in \eqref{eq::5_rely_on_n+k} for \(k\geq 2\) are as well. Because everything funnels through \eqref{eq::5_rely_on_n+k} and \eqref{eq::5_KdV_n+k}, we can obtain the whole solution \((n_{\pm}^{(k)},v_{\pm}^{(k)},\phi^{(k)})\) for \(k\geq 1\). 
A theorem and proof of this is as follows.

\begin{center}{\rule{4cm}{0.4pt}}\end{center}
\begin{theorem}\label{thm::5_solve_nkKdV}
    Let \(k\geq 2\) and \(s_k \leq s_1 - 3(k-1)\) be sufficiently large integers. 
    Then for any \(T>0\) and any given initial data \(\left(n_{\pm}^{(k)}(Y,0), v_{\pm}^{(k)}(Y,0), \phi^{(k)}(Y,0) \right)~\in~H^{s_k}(\mathbb{R})\) satisfying \eqref{eq::5_rely_on_n+k}, the initial value problem \eqref{eq::5_KdV_n+k} with \eqref{eq::5_rely_on_n+k} has a unique solution 
    \begin{equation*}\left(n_{\pm}^{(k)}, v_{\pm}^{(k)}, \phi^{(k)} \right)~\in L^\infty\left(H^{s_k}(\mathbb{R}) ; -T, T\right). 
    \end{equation*}
\end{theorem}
\begin{proof}
    Consider the equation \eqref{eq::5_KdV_n+k} for \(k\geq 2\). Let \(T>0\) be arbitrarily fixed and \(s_1\) be sufficiently large as in Thereom~\ref{thm::5_solve_n1KdV}.
    
    \textit{Existence}: Sach's 1983 paper, \cite{sachs1983completeness}, provides an explicit solution to the Cauchy problem for the linearized KdV, which is the left-hand side of \eqref{eq::5_KdV_n+k}. So, we have a solution for the homogeneous equation we can use to obtain a solution to our inhomogeneous problem by invoking Duhamel's Principle. Moreover, such a solution is asymptotically null. That is, \(\lim_{Y\to \pm\infty} \lvert n_+^{(k)}(Y,T) \rvert= 0\).
    
    \textit{Uniqueness}: Suppose that \(n_{+,1}^{(k)}\) and \(n_{+,2}^{(k)}\) are both solutions to \eqref{eq::5_KdV_n+k}. Call
    \begin{equation*}
        n^{(k)} \coloneqq n_{+,1}^{(k)} - n_{+,2}^{(k)}.
    \end{equation*}
    Then 
    \begin{equation*}\partial_{T} n^{(k)} + c\partial_{Y}\left(n_{+}^{(1)}n^{(k)}\right) + \frac{(c^2 m_e -1)^2}{2c(1+m_e)}\partial_{Y}^3 n^{(k)} = 0. 
    \end{equation*}
    Multiply this by \(n^{(k)}\) and integrate with respect to \(Y\) on \(\mathbb{R}\) to obtain 
    \begin{equation*}\int_{\mathbb{R}} n^{(k)}\partial_{T} n^{(k)} dY + c\int_{\mathbb{R}} n^{(k)}\partial_{Y}\left(n_{+}^{(1)}n^{(k)}\right) dY + \frac{(c^2 m_e -1)^2}{2c(1+m_e)}\int_{\mathbb{R}} n^{(k)}\partial_{Y}^3 n^{(k)} dY = 0. 
    \end{equation*} 
    Recall that \(n_{+,1}^{(k)}\) and \(n_{+,2}^{(k)}\) are asymptotically null. Hence, \(n^{(k)}\) is as well. Combining integration by parts along with this fact, the third integral vanishes in the limit. Therefore we have 
    \begin{equation*}
    \int_{\mathbb{R}} n^{(k)}\partial_{T} n^{(k)} dY + c\int_{\mathbb{R}} n^{(k)}\partial_{Y}\left(n_{+}^{(1)}n^{(k)}\right) dY = 0,
    \end{equation*} 
    or 
    \begin{equation*}
    \int_{\mathbb{R}} n^{(k)}\partial_{T} n^{(k)} dY + c\int_{\mathbb{R}} n_{+}^{(1)} n^{(k)}\partial_{Y}n^{(k)} dY + c\int_{\mathbb{R}} (n^{(k)})^2 \partial_{Y}n_{+}^{(1)} dY = 0.
    \end{equation*} 
    Define 
    \begin{equation*} 
    E^2(T) = \frac{1}{2}\lVert n^{(k)} \rVert_{H^1}^2.
    \end{equation*} 
    Then by differentiating both sides with respect to \(T\), we have \begin{equation*}
    \frac{dE^2(T)}{dT} = \frac{1}{2} \frac{d}{dT}\lVert n^{(k)} \rVert_{H^1}^2 = \int_{\mathbb{R}} n^{(k)} \partial_T n^{(k)} dY + \int_{\mathbb{R}} \partial_Y n^{(k)} \partial_Y\partial_T n^{(k)} dY.
    \end{equation*}
    On the first integral, we can apply \eqref{eq::5_KdV_n+k}. Observe 
    \begin{align*}
        \int n^{(k)} \partial_T n^{(k)} dY & = -c\int_{\mathbb{R}} n_{+}^{(1)} n^{(k)}\partial_{Y}n^{(k)} dY - c\int_{\mathbb{R}} (n^{(k)})^2 \partial_{Y}n_{+}^{(1)} dY \\[3pt]
        & \leq c \lVert n_{+}^{(1)} \rVert_{L^\infty} \int_{\mathbb{R}}  \partial_{Y} \left((n^{(k)})^2\right) dY + c \lVert \partial_{Y} n_{+}^{(1)} \rVert_{L^\infty} \lVert n^{(k)} \rVert_{L^2}^2 \\[3pt]
        & \leq \tilde{C} \lVert n^{(k)} \rVert_{H^1}^2
    \end{align*}
    For the second integral,
    \begin{align*}
        \int_{\mathbb{R}} \partial_Y n^{(k)} \partial_Y\partial_T n^{(k)} dY & = \int_{\mathbb{R}} \partial_Y n^{(k)} \left(-c\partial_{Y}^2\left(n_{+}^{(1)}n^{(k)}\right) - \frac{(c^2 m_e -1)^2}{2c(1+m_e)}\partial_{Y}^4 n^{(k)}\right) dY \\[3pt]
        & = -c \int_{\mathbb{R}} \partial_Y n^{(k)}\partial_{Y}^2\left(n_{+}^{(1)}n^{(k)}\right) dY - \frac{(c^2 m_e -1)^2}{2c(1+m_e)} \int_{\mathbb{R}} \partial_Y n^{(k)} \partial_{Y}^4 n^{(k)} dY. 
    \end{align*}
    The latter disappears in the limit through integration by parts. The former expands to
    \begin{align*}
        \int_{\mathbb{R}} \partial_Y n^{(k)} \partial_Y\partial_T n^{(k)} dY & = -c \int_{\mathbb{R}} n_{+}^{(1)} \partial_Y n^{(k)} \partial_Y^2 n^{(k)} dY - 2c \int_{\mathbb{R}} \partial_Y n_{+}^{(1)} (\partial_Y n^{(k)})^2 dY \\
        & \quad - c \int_{\mathbb{R}} \partial_Y^2 n_{+}^{(1)} n^{(k)} \partial_Y n^{(k)} dY. 
    \end{align*}
    After applying integration by parts, we can bound this integral as
    \begin{equation*}
        \int_{\mathbb{R}} \partial_Y n^{(k)} \partial_Y\partial_T n^{(k)} dY \leq 2c \lVert \partial_Y n_{+}^{(1)} \rVert_{L^\infty} \lVert \partial_Y n^{(k)} \rVert_{L^2}^2 \leq \tilde{C} \lVert n^{(k)} \rVert_{H^1}^2.
    \end{equation*}
    Thus 
    \begin{equation*}
    \frac{dE^2(T)}{dT} \leq \tilde{C} E^2(T) \quad\implies\quad \dot{E}(T) \leq \tilde{C} E(T).
    \end{equation*} 
    We can use an integrating factor to find 
    \begin{equation*}
    E(T) \leq E(0)e^{\tilde{C}T}.
    \end{equation*} 
    Since \(n_{+,1}^{(k)}\) and \(n_{+,2}^{(k)}\) are both solutions to \eqref{eq::5_KdV_n+k}, they have the same initial condition. Hence, \(n^{(k)}(Y,0) = 0\) and ergo, \(E(0)=0\). This shows that \(E(T) = 0 \). Therefore, \(n_{+,1}^{(k)} \equiv n_{+,2}^{(k)} \) for all finite \(T>0\), so the solution to \eqref{eq::5_KdV_n+k} is unique.
    
    \textit{Boundedness}: Since \(F^{(k-1)}\) only depends on \(n_+^{(i-1)} \in H^{s_i}(\mathbb{R})\) for \(1\leq i \leq k-1\), we know it is bounded in \(L^{\infty}(H^{s_{k-1}}; -T,T)\) for all \(T\geq 0\). Now, multiply \eqref{eq::5_KdV_n+k} by \(n_+^{(k)}\). We will get 
    \begin{equation*}
        n_+^{(k)}\partial_{T} n_{+}^{(k)} + cn_+^{(k)}\partial_{Y}\left(n_{+}^{(1)}n_{+}^{(k)}\right) + \frac{(c^2 m_e -1)^2}{2c(1+m_e)}n_+^{(k)} \partial_{Y}^3 n_{+}^{(k)} = n_+^{(k)}F^{(k-1)}
    \end{equation*}
    Integrating with respect to \(Y\) over \(\mathbb{R}\), we see
    \begin{align*}
        \frac{d}{dT}\bigg[\frac{1}{2} \lVert n_+^{(k)} \rVert_{L^2}^2\bigg] & \leq c \int_{\mathbb{R}} n_{+}^{(1)}n_+^{(k)}\partial_{Y}n_+^{(k)} dY + c \int_{\mathbb{R}} (n_+^{(k)})^2 \partial_{Y}n_+^{(1)} dY \\[3pt]
        & \quad + \frac{(c^2 m_e -1)^2}{2c(1+m_e)} \int_{\mathbb{R}} n_+^{(k)} \partial_{Y}^3 n_{+}^{(k)} dY + \int_{\mathbb{R}} n_+^{(k)}F^{(k-1)} dY.
    \end{align*}
    We can apply integration by parts to find many of the terms vanish in the limit, and we bound the remaining as
    \begin{align*}
        \frac{d}{dT}\bigg[\frac{1}{2} \lVert n_+^{(k)} \rVert_{L^2}^2\bigg] & \leq \frac{c}{2}\int_{\mathbb{R}} (n_+^{(k)})^2 \partial_{Y}n_+^{(1)} dY + \int_{\mathbb{R}} n_+^{(k)}F^{(k-1)} dY \\[10pt]
        & \leq \bar{C} \lVert \partial_{Y}n_+^{(1)} \rVert_{L^\infty} \lVert n_+^{(k)} \rVert_{L^2}^2 + \lVert F^{(k-1)} \rVert_{L^\infty} \lVert n_+^{(k)} \rVert_{L^2} \\[3pt]
        & \leq \bar{C} \lVert n_+^{(k)} \rVert_{L^2}^2 + \bar{C} \lVert n_+^{(k)} \rVert_{L^2}.
    \end{align*}
    Along this same vein, we can take \(\partial_{Y}^{s_k}\) of \eqref{eq::5_KdV_n+k}, and multiply it by \(\partial_{Y}^{s_k} n_+^{(k)}\) to find 
    \begin{equation*}\frac{d}{dT}\bigg[\frac{1}{2} \lVert n_+^{(k)} \rVert_{H^{s_k}}^2\bigg] \leq \bar{C} \lVert n_+^{(k)} \rVert_{H^{s_k}}^2 + \bar{C} \lVert n_+^{(k)} \rVert_{H^{s_k}}. 
    \end{equation*}
    Now, define 
    \begin{equation*}
    \eta^2(T) \coloneqq \frac{1}{2} \lVert n_+^{(k)} \rVert_{H^{s_k}}^2 .
    \end{equation*}
    Then \(\eta(0)\geq 0\), and by our previous inequality, 
    \begin{equation*}
    \dot{\eta}(T) \leq \bar{C}\eta(T) + \bar{C}. 
    \end{equation*} 
    We will show \(\exists \hat{C} > 0\) such that 
    \begin{equation*} 
    \sup_{0\leq T \leq T_{*}} \eta (T) \leq \hat C 
    \end{equation*} 
    for a fixed \(T_{*} >0\). 
    
    So, let \(\theta = \gamma(e^{\bar{C} T_{*}}(\eta(0) + 1) - 1) > 0\) with \(\gamma > 1\) and define
    \begin{equation*}
    T_{**} = \inf \{ S \, | \, \sup_{0\leq T \leq S} \eta(T) \leq \theta \}
    \end{equation*} 
    Then either \(T_{**} > T_{*}\) or \(T_{**} \leq T_{*}\). If the former, then we are done with \(\hat C = \theta\) since \begin{equation*}\sup_{0\leq T \leq T_{*}} \eta(T) \leq \sup_{0\leq T \leq T_{**}} \eta(T) \leq \theta
    \end{equation*}
    A specific value of \(\theta\) plays no role here. If the latter, under the assumption that \(\eta(T)\) is continuous, we know 
    \(
    \eta(T_{**}) = \theta. 
    \)
    For \(0 \leq T \leq T_{**}\), we can use our handy-dandy integrating factor \(e^{-\bar{C}T}\) to get
    \begin{equation*}\dot{\eta}(T)e^{-\bar{C}T} - \bar{C}\eta(T)e^{-\bar{C}T} \leq \bar{C}e^{-\bar{C}T}.
    \end{equation*} 
    This is the same as 
    \begin{equation*}\frac{d}{dT} \left( \eta(T)e^{-\bar{C}T}\right) \leq \bar{C}e^{-\bar{C}T}. \end{equation*} Integrating both sides with respect to \(T\) from \(0\) to \(\tau\) we get 
    \(
    \eta(\tau)e^{-\bar{C} \tau} - \eta(0) \leq 1 - e^{-\bar{C}\tau}.
    \)
    We will do some rewriting to obtain
    \begin{equation*}
    \eta(T) \leq \eta(0)e^{\bar{C} T} + e^{\bar{C} T} - 1.
    \end{equation*} 
    Since \(T \leq T_{**} \leq T_{*}\), this is
    \begin{equation*}
    \eta(T) \leq \eta(0)e^{\bar{C} T_{*}} + e^{\bar{C} T_{*}} - 1 .
    \end{equation*} 
    In particular, this estimate holds at \(T_{**}\). Hence,
    \begin{equation*}
    \theta = \eta(T_{**}) \leq e^{\bar{C} T_{*}}(\eta(0) + 1) - 1.
    \end{equation*}
    But this means 
    \(
        \theta \leq \frac{\theta}{\gamma}
    \)
    with  \(\gamma > 1 \) -- a clear contradiction! Hence, we have the boundedness,
    \begin{equation*}
    \sup_{0\leq T \leq T_{*}} \lVert n_+^{(k)}(\cdot, T) \rVert_{H^{s_k}} \leq \hat{C}\left(\lVert n_+^{(k)}(\cdot,0) \rVert_{H^{s_k}}\right).
    \end{equation*} 
    Moreover, \(n_+^{(k)}\) is a global solution as conservation laws of the KdV equation allow us to extend our bound to any interval \([-T,T]\). 
\end{proof}

\begin{center}{\rule{4cm}{0.4pt}}\end{center}

We have now shown that we can find a solution for each \(k^{\text{th}}\) step of our formal expansion \eqref{eq::form_exp}. As \(k\) increases, the orders of \(\varepsilon\) become larger, and the impact of \(\left(n_{\pm}^{(k)}, v_{\pm}^{(k)}, \phi^{(k)}\right)\) on the overall expansion diminishes. Essentially, these minuscule pieces can be combined. 

\input{kdv_residual.tex}

%% file: kdv_residual.tex
\subsection{Bounding the Residual Terms}\label{5_kdv_residual}

We subsequently break up the formal expansion to include a remainder segment that encompasses the higher-order \(\varepsilon\) terms. This new formulation is
\begin{align*}
    n_{\pm} & = 1 + \sum_{k=1}^{m} \varepsilon^k n_{\pm}^{(k)}\left(\varepsilon^{1/2}(x-ct),\varepsilon^{3/2}t\right) + \varepsilon^{m-1} n_{\pm}^{(R_\varepsilon)}(x,t) \coloneqq N_{\pm}^{(m)}(Y,T) + \varepsilon^{m-1} n_{\pm}^{(R_\varepsilon)}(x,t), \\[2pt]
    v_{\pm} & = \phantom{1 + } \sum_{k=1}^{m} \varepsilon^k v_{\pm}^{(k)}\left(\varepsilon^{1/2}(x-ct),\varepsilon^{3/2}t\right) + \varepsilon^{m-1} v_{\pm}^{(R_\varepsilon)}(x,t) \coloneqq V_{\pm}^{(m)}(Y,T) + \varepsilon^{m-1} v_{\pm}^{(R_\varepsilon)}(x,t), \\[2pt]
    \phi & = \phantom{1 + } \sum_{k=1}^{m} \varepsilon^k \phi^{(k)}\left(\varepsilon^{1/2}(x-ct),\varepsilon^{3/2}t\right) + \varepsilon^{m-1} \phi^{(R_\varepsilon)}(x,t) \coloneqq \Phi^{(m)}(Y,T) + \varepsilon^{m-1} \phi^{(R_\varepsilon)}(x,t), 
\end{align*}
where \(\left(n_{\pm}^{(k)}, v_{\pm}^{(k)}, \phi^{(k)} \right)\) satisfy \eqref{eq::5_rely_on_n+k} and \eqref{eq::5_KdV_n+k} for \(1\leq k \leq m\). We plug into \eqref{eq::5hotrescale} to get
\begin{subequations}\label{eq::5_remainder_sys}
\begin{empheq}[left = (\mathscr{R}_m) \empheqlbrace\,]{align}
    \label{eq::npm_remainder}
    \partial_t n_{\pm}^{(R_\varepsilon)} + \partial_x\left[N_{\pm}^{(m)}v_{\pm}^{(R_\varepsilon)} + V_{\pm}^{(m)}n_{\pm}^{(R_\varepsilon)} + \varepsilon^{m-1}n_{\pm}^{(R_\varepsilon)}v_{\pm}^{(R_\varepsilon)}\right] & = \text{Res}_{n_\pm}^{(m)}, \\[4pt]
    \label{eq::vp_remainder}
    \nonumber \partial_t v_{+}^{(R_\varepsilon)} + \partial_x \left[v_{+}^{(R_\varepsilon)}V_{+}^{(m)} + \tfrac{\varepsilon^{m-1}}{2}(v_{+}^{(R_\varepsilon)})^2\right] \\
    + \partial_x \left[\tau_i \ln(N_{+}^{(m)} + \varepsilon^{m-1}n_{+}^{(R_\varepsilon)}) + \phi^{(R_\varepsilon)} \right] & = \text{Res}_{v_+}^{(m)}, \\[4pt]
    \label{eq::vm_remainder}
    \nonumber m_e\left(\partial_t v_{-}^{(R_\varepsilon)} + \partial_x \left[v_{-}^{(R_\varepsilon)}V_{-}^{(m)} + \tfrac{\varepsilon^{m-1}}{2}(v_{-}^{(R_\varepsilon)})^2\right]\right) \\
    \qquad + \partial_x\left[ \ln\left(N_{-}^{(m)} + \varepsilon^{m-1}n_{-}^{(R_\varepsilon)}\right) - \phi^{(R_\varepsilon)} \right] & = \text{Res}_{v_-}^{(m)}, \\[4pt]
    \label{eq::phi_remainder}
    \partial_x^2 \phi^{(R_\varepsilon)} - n_{-}^{(R_\varepsilon)} + n_{+}^{(R_\varepsilon)} & = \text{Res}_{\phi}^{(m)},
\end{empheq}
\end{subequations}
the Remainder System which contains terms,
\begin{align*}
    \text{Res}_{n_\pm}^{(m)} & = -\varepsilon^{1-m}\Big[\partial_t N_{\pm}^{(m)} + \partial_x (N_{\pm}^{(m)}V_{\pm}^{(m)})\Big], \\[4pt]
    \text{Res}_{v_+}^{(m)} & = -\varepsilon^{1-m}\Big[\partial_t V_{+}^{(m)} + \tfrac{1}{2}\partial_x (V_{+}^{(m)})^2 + \tau_i\sum_{k=1}^{m} \frac{(-1)^{k+1}}{k}\partial_x\left((N_{+}^{(m)} - 1)^k\right) + \partial_x \Phi^{(m)}\Big], \\[4pt]
    \text{Res}_{v_-}^{(m)} & = -\varepsilon^{1-m}\Big[m_e\left(\partial_t V_{-}^{(m)} + \tfrac{1}{2}\partial_x (V_{-}^{(m)})^2\right) + \sum_{k=1}^{m} \frac{(-1)^{k+1}}{k}\partial_x\left((N_{-}^{(m)} - 1)^k\right) - \partial_x \Phi^{(m)}\Big], \\[4pt]
    \text{Res}_{\phi_{\phantom{+}}}^{(m)} & = -\varepsilon^{1-m}\Big[\partial_x^2 \Phi^{(m)} - N_{-}^{(m)} + N_{+}^{(m)}\Big].
\end{align*}

They are the residual quantity, the amount the initial estimate fails to solve the system, scaled by the expected error \(\varepsilon^{1-m}\). To study this error, we will start by measuring the size of these terms. Seemingly, we run into issues when applying values of \(m>1\), as \(\varepsilon^{1-m}\) blows up. However, not all hope is lost because these terms can be rewritten with a power of \(\varepsilon\) that does not depend on the expansion limit, \(m\).

\begin{lemma}\label{lem::res_dens}
    We can rewrite the density error-scaled residual terms as 
    \begin{align*}
    \text{Res}_{n_\pm}^{(m)} = -\varepsilon^{5/2} \left(\partial_T n_\pm^{(m)} + \sum_{j=1}^{m} \sum_{k=j}^{m} \varepsilon^{j-1} \partial_Y\left(n_{\pm}^{(k)} v_{\pm}^{(m+1+j-k-1)}\right)\right).
    \end{align*}
\end{lemma}

\begin{proof}
We will proceed by induction. For the base case \(m=1\),
\begin{align*}\text{Res}_{n_\pm}^{(1)} = -\Big[\partial_t N_{\pm}^{(1)} + \partial_x (N_{\pm}^{(1)}V_{\pm}^{(1)})\Big].\end{align*} Applying the Gardner-Morikawa transformation, the residuals become
\begin{align*}\text{Res}_{n_\pm}^{(1)} = -\varepsilon^{1/2}\Big[-c \partial_Y N_{\pm}^{(1)} + \varepsilon \partial_T N_{\pm}^{(1)} + \partial_Y (N_{\pm}^{(1)}V_{\pm}^{(1)})\Big].\end{align*}
Expanding all terms gives
\begin{align*}
    \text{Res}_{n_\pm}^{(1)} & = -\varepsilon^{1/2}\Big[-c \partial_Y (\varepsilon n_{\pm}^{(1)}) + \varepsilon \partial_T (\varepsilon n_{\pm}^{(1)}) + \partial_Y ((1 + \varepsilon n_{\pm}^{(1)})(\varepsilon v_{\pm}^{(1)}))\Big] \\[4pt]
    & = -\varepsilon^{3/2}\Big[-c \partial_Y n_{\pm}^{(1)} + \varepsilon \partial_T n_{\pm}^{(1)} + \partial_Y v_{\pm}^{(1)} +  \varepsilon \partial_Y(n_{\pm}^{(1)} v_{\pm}^{(1)})\Big].
\end{align*}
We will use \eqref{eq::rely_on_n+1} to replace \(v_\pm^{(1)} = cn_\pm^{(1)}\), and obtain
\begin{align*}
    \text{Res}_{n_\pm}^{(1)} & = -\varepsilon^{3/2}\Big[-c \partial_Y n_{\pm}^{(1)} + \varepsilon \partial_T n_{\pm}^{(1)} + c\partial_Y n_{\pm}^{(1)} +  \varepsilon \partial_Y(n_{\pm}^{(1)} v_{\pm}^{(1)})\Big] \\[4pt]
    & = -\varepsilon^{5/2}\Big[\partial_T n_{\pm}^{(1)} + \partial_Y(n_{\pm}^{(1)} v_{\pm}^{(1)})\Big].
\end{align*}

We now make the induction hypothesis
\begin{equation*}
    \text{Res}_{n_\pm}^{(m-1)} = -\varepsilon^{5/2} \left(\partial_T n_\pm^{(m-1)} + \sum_{j=1}^{m-1} \sum_{k=j}^{m-1} \varepsilon^{j-1} \partial_Y\left(n_{\pm}^{(k)} v_{\pm}^{(m+j-k-1)}\right)\right).
\end{equation*}
On the \(m^{\text{th}}\) step, the density residuals are \begin{align*}
\text{Res}_{n_\pm}^{(m)} = -\varepsilon^{1-m}\Big[\partial_t N_{\pm}^{(m)} + \partial_x (N_{\pm}^{(m)}V_{\pm}^{(m)})\Big],
\end{align*} 
which we can rewrite with the Gardner-Morikawa transformation. Following a separation that extracts the \(m^{\text{th}}\) term, e.g. \(N_\pm^{(m)} = N_{\pm}^{(m-1)} + \varepsilon^m n_+^{(m)}\), it becomes
\begin{align*}
    \text{Res}_{n_\pm}^{(m)} = -\varepsilon^{3/2-m}\Big[& -c \partial_Y N_{\pm}^{(m-1)} + \varepsilon \partial_T N_{\pm}^{(m-1)} + \partial_Y \left( N_{\pm}^{(m-1)}V_{\pm}^{(m-1)}\right) - c \varepsilon^m \partial_Y n_{\pm}^{(m)} \\
    & + \varepsilon^{m+1} \partial_T n_{\pm}^{(m)} + \varepsilon^m \partial_Y \left((N_{\pm}^{(m-1)} - 1) v_{\pm}^{(m)} + V_{\pm}^{(m-1)} n_{\pm}^{(m)} + \varepsilon^m n_{\pm}^{(m)}v_{\pm}^{(m)} \right) \\
    & + \varepsilon^m \partial_Y v_{\pm}^{(m)} \Big].
\end{align*}
Notably, this is the same as
\begin{align*}
    \text{Res}_{n_\pm}^{(m)} & = \varepsilon^{-1} \text{Res}_{n_\pm}^{(m-1)} \\
    & \phantom{=} -\varepsilon^{3/2}\Big[-c\partial_Y n_{\pm}^{(m)} + \varepsilon \partial_T n_{\pm}^{(m)} + \partial_Y \left((N_{\pm}^{(m-1)} - 1) v_{\pm}^{(m)} + V_{\pm}^{(m-1)} n_{\pm}^{(m)} + \varepsilon^m n_{\pm}^{(m)}v_{\pm}^{(m)} \right) \\
    & \phantom{= -\varepsilon^{3/2}\Big[} + \partial_Y v_{\pm}^{(m)} \Big].
\end{align*}
Replacing the last \(v_{\pm}^{(m)}\) with its form from \eqref{eq::5_rely_on_n+k} and \eqref{eq::5_rely_on_n+k_spec}, we have
\begin{align*}
    \text{Res}_{n_\pm}^{(m)} & = \varepsilon^{-1} \text{Res}_{n_\pm}^{(m-1)} \\
    & \phantom{=} -\varepsilon^{3/2}\Bigg[-c\partial_Y n_{\pm}^{(m)} + \varepsilon \partial_T n_{\pm}^{(m)} + \partial_Y \left((N_{\pm}^{(m-1)} - 1) v_{\pm}^{(m)} + V_{\pm}^{(m-1)} n_{\pm}^{(m)} + \varepsilon^m n_{\pm}^{(m)}v_{\pm}^{(m)} \right) \\
    & \phantom{= -\varepsilon^{3/2}\Big[} + c \partial_Y n_{\pm}^{(m)} - \partial_{T}n_{\pm}^{(m-1)} - \partial_Y \left(\sum_{k=1}^{m-1} n_{\pm}^{(k)} v_{\pm}^{(m-k)}\right) \Bigg].
\end{align*}
Inserting the induction hypothesis yields
\begin{align*}
    \text{Res}_{n_\pm}^{(m)} & = -\varepsilon^{3/2} \Bigg[\partial_T n_\pm^{(m-1)} + \sum_{j=1}^{m-1} \sum_{k=j}^{m-1} \varepsilon^{j-1} \partial_Y\left(n_{\pm}^{(k)} v_{\pm}^{(m+j-k-1)}\right)\Bigg]  \\
    & \quad -\varepsilon^{3/2}\Bigg[\varepsilon \partial_T n_{\pm}^{(m)} + \partial_Y \left((N_{\pm}^{(m-1)} - 1) v_{\pm}^{(m)} + V_{\pm}^{(m-1)} n_{\pm}^{(m)} + \varepsilon^m n_{\pm}^{(m)}v_{\pm}^{(m)} \right) \\
    & \phantom{= -\varepsilon^{3/2}\Bigg[} - \partial_{T}n_{\pm}^{(m-1)} - \partial_Y\left(\sum_{k=1}^{m-1} n_{\pm}^{(k)} v_{\pm}^{(m-k)}\right) \Bigg].
\end{align*}
Now, we simplify and expand all variables to reveal
\begin{align*}
    \text{Res}_{n_\pm}^{(m)} & = -\varepsilon^{3/2} \Bigg[\sum_{j=2}^{m-1} \sum_{k=j}^{m-1} \varepsilon^{j-1} \partial_Y\left(n_{\pm}^{(k)} v_{\pm}^{(m+j-k-1)}\right) \\
    & \phantom{= -\varepsilon^{3/2}\Big[} + \varepsilon \partial_T n_{\pm}^{(m)} + \partial_Y \left(\sum_{k=1}^{m-1} \varepsilon^k n_\pm^{(k)}v_{\pm}^{(m)} + \sum_{k=1}^{m-1} \varepsilon^k v_\pm^{(k)}n_{\pm}^{(m)} + \varepsilon^m n_{\pm}^{(m)}v_{\pm}^{(m)}  \right) \Bigg].
\end{align*}
Finally, we can factor out an \(\varepsilon\) and combine the summations, resulting in 
\begin{align*}
    \text{Res}_{n_\pm}^{(m)} & = -\varepsilon^{5/2} \Bigg[\partial_T n_\pm^{(m)} + \sum_{j=1}^{m} \sum_{k=j}^{m} \varepsilon^{j-1} \partial_Y\left(n_{\pm}^{(k)} v_{\pm}^{(m+1+j-k-1)}\right)\Bigg].
\end{align*}
For a more in-depth calculation that equates the last two steps, consider
\begin{align*}
    & \sum_{j=1}^{m} \sum_{k=j}^{m} \varepsilon^{j-1} n_{\pm}^{(k)} v_{\pm}^{(m+1+j-k-1)} \\[2pt] 
    & = \sum_{j=1}^{m-1} \sum_{k=j}^{m} \varepsilon^{j-1} n_{\pm}^{(k)} v_{\pm}^{(m+j-k)} + \varepsilon^{m-1} n_{\pm}^{(m)} v_{\pm}^{(m)} \\[2pt] 
    & = \sum_{j=1}^{m-1} \sum_{k=j}^{m-1} \varepsilon^{j-1} n_{\pm}^{(k)} v_{\pm}^{(m+j-k)} + \sum_{j=1}^{m-1} \varepsilon^{j-1} n_{\pm}^{(m)} v_{\pm}^{(j)} + \varepsilon^{m-1} n_{\pm}^{(m)} v_{\pm}^{(m)} \\[2pt] 
    & = \sum_{j=1}^{m-1} \sum_{k=j+1}^{m-1} \varepsilon^{j-1} n_{\pm}^{(k)} v_{\pm}^{(m+j-k)} + \sum_{j=1}^{m-1} \varepsilon^{j-1} n_{\pm}^{(j)} v_{\pm}^{(m)} + \sum_{j=1}^{m-1} \varepsilon^{j-1} n_{\pm}^{(m)} v_{\pm}^{(j)} + \varepsilon^{m-1} n_{\pm}^{(m)} v_{\pm}^{(m)}  \\[2pt] 
    & = \sum_{j=2}^{m-1} \sum_{k=j}^{m-1} \varepsilon^{j-1} n_{\pm}^{(k)} v_{\pm}^{(m+j-k-1)} + \sum_{k=1}^{m-1} \varepsilon^{k-1}\left( n_\pm^{(k)}v_{\pm}^{(m)} + v_\pm^{(k)}n_{\pm}^{(m)}\right) + \varepsilon^{m-1} n_{\pm}^{(m)}v_{\pm}^{(m)}.
\end{align*}
\end{proof}

\begin{lemma}\label{lem::res_vel_ion}
    We can rewrite the ion velocity error-scaled residual terms as 
    \begin{align*}
        & \text{Res}_{v_+}^{(m)} = -\varepsilon^{5/2}\left[\rule{0pt}{28pt}\right. \partial_T v_+^{(m)} + \tfrac{1}{2}\partial_Y\left(\sum_{k=1}^{m} \sum_{j=k}^{m} \varepsilon^{k-1} v_{+}^{(j)} v_{+}^{(m+k-j)} \right) \\[8pt] 
        & + \tau_i \partial_Y\left( \sum_{k=2}^{m}\frac{(-1)^{k+1}\varepsilon^{k-m-1}}{k}  \left(\sum_{\substack{S_{m-1} = k \\ P_{m-1} \geq m+1-k}} \binom{k}{a_1,a_2,\ldots,a_{m-1}} \prod_{i=1}^{m-1} \left(\varepsilon^{i-1}n_+^{(i)} \right)^{a_i}\right) \right) \\[8pt] 
        & + \tau_i\partial_Y \left(\sum_{k=2}^{m} \frac{(-1)^{k+1}}{k} \left(\sum_{j=1}^{k} \binom{k}{j} \varepsilon^{m(j-1) + k - j - 1} \left(\sum_{i=1}^{m-1} \varepsilon^{i-1} n_+^{(i)} \right)^{k-j} \left( n_{+}^{(m)}\right)^{j}\right)\right) \left.\rule{0pt}{28pt}\right],
    \end{align*}
    where \(0 \leq a_i \in\mathbb{Z}\) for all \(1\leq i \leq m-1\) and
    \begin{equation*}
    S_{m-1} = \sum_{i=1}^{m-1} a_i, \quad P_{m-1} = \sum_{i=1}^{m-1} (i-1) a_i.
    \end{equation*}
\end{lemma}

\begin{proof}
For the base case, we know
\begin{equation*}
\text{Res}_{v_+}^{(1)} = - \Big[\partial_t V_{+}^{(1)} + \tfrac{1}{2}\partial_x (V_{+}^{(1)})^2 + \tau_i\sum_{k=1}^{1} \frac{(-1)^{k+1}}{k}\partial_x\left((N_{+}^{(1)} - 1)^k\right) + \partial_x \Phi^{(1)}\Big].
\end{equation*}
By simplifying and applying the Gardner-Morikawa transformation, it becomes,
\begin{equation*}
\text{Res}_{v_+}^{(1)} = -\varepsilon^{1/2}\Big[-c\partial_Y V_{+}^{(1)} + \varepsilon\partial_T V_{+}^{(1)} + \tfrac{1}{2}\partial_Y (V_{+}^{(1)})^2 + \tau_i \partial_Y N_{+}^{(1)} + \partial_Y \Phi^{(1)}\Big].
\end{equation*}
Expanding the variables, we get 
\begin{equation*}
\text{Res}_{v_+}^{(1)} = -\varepsilon^{3/2}\Big[-c\partial_Y v_{+}^{(1)} + \varepsilon\partial_T v_{+}^{(1)} + \varepsilon v_{+}^{(1)}\partial_Y v_{+}^{(1)} + \tau_i \partial_Y n_{+}^{(1)} + \partial_Y \phi^{(1)}\Big].
\end{equation*}
From \eqref{eq::rely_on_n+1}, we can replace \(\phi^{(1)}\) with its relationship to \(n_+^{(1)}\) and obtain 
\begin{equation*}
\text{Res}_{v_+}^{(1)} = -\varepsilon^{3/2}\Big[-c\partial_Y v_{+}^{(1)} + \varepsilon\partial_T v_{+}^{(1)} + \varepsilon v_{+}^{(1)}\partial_Y v_{+}^{(1)} + \tau_i\partial_Y n_{+}^{(1)} + c \partial_Y v_{+}^{(1)} - \tau_i \partial_Y n_{+}^{(1)}\Big].
\end{equation*}
Ergo, 
\begin{equation*}
\text{Res}_{v_+}^{(1)} = -\varepsilon^{5/2}\Big[\partial_T v_{+}^{(1)} +  v_{+}^{(1)}\partial_Y v_{+}^{(1)} \Big].
\end{equation*}

Coming out of left field, we now make the most insane induction hypothesis,
\begin{align*}
    & \text{Res}_{v_+}^{(m-1)} = -\varepsilon^{5/2}\left[\rule{0pt}{28pt}\right. \partial_T v_+^{(m-1)} + \tfrac{1}{2}\partial_Y\left(\sum_{k=1}^{m-1} \sum_{j=k}^{m-1} \varepsilon^{k-1} v_{+}^{(j)} v_{+}^{(m-1+k-j)} \right) \\
    & + \tau_i \partial_Y\left( \sum_{k=2}^{m-1}\frac{(-1)^{k+1}}{k}  \varepsilon^{k-m}\left(\sum_{\substack{S_{m-2} = k \\ P_{m-2} \geq m-k}} \binom{k}{a_1,a_2,\ldots,a_{m-2}} \prod_{i=1}^{m-2} \left(\varepsilon^{i-1}n_+^{(i)} \right)^{a_i}\right) \right) \\
    & + \tau_i\partial_Y \left(\sum_{k=2}^{m-1} \frac{(-1)^{k+1}}{k} \left(\sum_{j=1}^{k} \binom{k}{j} \varepsilon^{(m-1)(j-1) + k - j - 1} \left(\sum_{i=1}^{m-2} \varepsilon^{i-1} n_+^{(i)} \right)^{k-j} \left( n_{+}^{(m-1)}\right)^{j}\right)\right) \left.\rule{0pt}{28pt}\right],
\end{align*}
where \(0 \leq a_i\in\mathbb{Z}\) for all \(1\leq i \leq m-2\) and \[S_{m-2} = \sum_{i=1}^{m-2} a_i, \quad P_{m-2} = \sum_{i=1}^{m-2} (i-1) a_i.\]
\vspace{-8pt}
The \(m^\text{th}\) step of the ion velocity residual has the form
\begin{align*}
\text{Res}_{v_+}^{(m)} = -\varepsilon^{1-m}\Bigg[\partial_t V_{+}^{(m)} + V_{+}^{(m)} \partial_x V_{+}^{(m)} + \tau_i\partial_x \left(\sum_{k=1}^{m} \frac{(-1)^{k+1}}{k} (N_{+}^{(m)} - 1)^{k}\right) + \partial_x \Phi^{(m)}\Bigg].
\end{align*}
After applying the Gardner-Morikawa transformation, it becomes 
\begin{align*}
\text{Res}_{v_+}^{(m)} = -\varepsilon^{3/2-m}\Bigg[& -c\partial_Y V_{+}^{(m)} + \varepsilon \partial_T V_{+}^{(m)} + V_{+}^{(m)} \partial_Y V_{+}^{(m)} \\
& + \tau_i\partial_Y \left(\sum_{k=1}^{m} \frac{(-1)^{k+1}}{k} (N_{+}^{(m)} - 1)^{k}\right) + \partial_Y \Phi^{(m)}\Bigg].
\end{align*}
Detaching the lower order residual terms from its final term, e.g. \(V_+^{(m)} = V_{+}^{(m-1)} + \varepsilon^m v_+^{(m)}\), we have
\begin{align*}
\text{Res}_{v_+}^{(m)} & = \varepsilon^{-1} \text{Res}_{v_+}^{(m-1)} - \varepsilon^{3/2}\Big[\zeta_1 + \zeta_2 + \zeta_3 \Big], \\[8pt]
\zeta_1 & \coloneqq - c\partial_Y v_{+}^{(m)} + \varepsilon \partial_T v_{+}^{(m)} + \partial_Y\left(V_{+}^{(m-1)}  v_{+}^{(m)}\right) + \varepsilon^{m} v_{+}^{(m)}\partial_Y v_{+}^{(m)} + \partial_Y \phi_+^{(m)}, \\[8pt]
\zeta_2 & \coloneqq \tau_i (-1)^{m+1} \left(\sum_{k=1}^{m-1}\varepsilon^{k-1} n_+^{(k)}\right)^{m-1}\partial_Y \left(\sum_{k=1}^{m-1}\varepsilon^{k-1} n_+^{(k)}\right), \\[8pt]
\zeta_3 & \coloneqq \tau_i \partial_Y \left(\sum_{k=1}^{m} \frac{(-1)^{k+1}}{k} \sum_{j=1}^{k} \binom{k}{j} (N_{+}^{(m-1)} - 1)^{k-j} \varepsilon^{m(j-1)}\left( n_{+}^{(m)} \right)^j \right).
\end{align*}

These pieces can be simplified. We will begin with \(\zeta_1\) and replace \(\phi_+^{(m)}\) with its definition from \eqref{eq::5_rely_on_n+k} and \eqref{eq::5_rely_on_n+k_spec}. This yields
\begin{align*}
    \zeta_1 = & \; \varepsilon \partial_T v_{+}^{(m)} + \partial_Y\left(V_{+}^{(m-1)}  v_{+}^{(m)}\right) + \varepsilon^{m} v_{+}^{(m)}\partial_Y v_{+}^{(m)} \\[2pt]
    & - \partial_T v_{+}^{(m-1)} - \tfrac{1}{2}\partial_Y \left(\sum_{j=1}^{m-1} v_{+}^{(j)} v_{+}^{(m-j)} \right) - \tau_i \partial_Y \ell_+^{(m)}.
\end{align*}
We can expand the terms and exchange \(\ell_+^{(m)}\) with its \eqref{eq::ell_def} equivalent,
\begin{align*}
    \zeta_1 = \; & \varepsilon \partial_T v_{+}^{(m)} + \varepsilon \partial_Y \left( v_{+}^{(m)} \sum_{j=1}^{m-1} \varepsilon^{j-1} v_{+}^{(j)} \right) + \varepsilon^{m} v_{+}^{(m)}\partial_Y v_{+}^{(m)} \\
    & - \partial_T v_{+}^{(m-1)} - \tfrac{1}{2}\partial_Y \left(\sum_{j=1}^{m-1} v_{+}^{(j)} v_{+}^{(m-j)} \right) \\
        & - \tau_i \partial_Y \left( \sum_{k=1}^{m}\frac{(-1)^{k+1}}{k}\left(\sum_{\substack{a_1 + a_2 + \cdots + a_{m} = k \\ a_1 + 2a_2 + \cdots + m a_{m} = m \\ a_1,a_2,\ldots,a_{m} \geq 0}} \binom{k}{a_1,a_2,\ldots,a_{m}} \prod_{i=1}^{m} \left(n_+^{(i)}\right)^{a_i}\right)\right).
\end{align*}
Extracting the \(a_1 = m\) and \(a_m = 1\) terms, we get
\begin{align*}
    \zeta_1 = \; & \varepsilon \partial_T v_{+}^{(m)} + \varepsilon \partial_Y \left( v_{+}^{(m)} \sum_{j=1}^{m-1} \varepsilon^{j-1} v_{+}^{(j)} \right) + \varepsilon^{m} v_{+}^{(m)}\partial_Y v_{+}^{(m)} \\ 
    &  - \partial_T v_{+}^{(m-1)} - \tfrac{1}{2}\partial_Y \left(\sum_{j=1}^{m-1} v_{+}^{(j)} v_{+}^{(m-j)} \right) - \tau_i \partial_Y\left( \frac{(-1)^{m+1}}{m}\left(n_+^{(1)}\right)^{m} + n_+^{(m)}\right) \\ 
    & - \tau_i \partial_Y \left( \sum_{k=1}^{m}\frac{(-1)^{k+1}}{k}\left(\sum_{\substack{a_1 + a_2 + \cdots + a_{m} = k \\ a_1 + 2a_2 + \cdots + m a_{m} = m \\ a_1,a_2,\ldots,a_{m} \geq 0 \\ a_1 \neq m \;\&\; a_m \neq 1}} \binom{k}{a_1,a_2,\ldots,a_{m}} \prod_{i=1}^{m} \left(n_+^{(i)}\right)^{a_i}\right)\right).
\end{align*}
This is how we will leave \(\zeta_1\). Next, we can rewrite \(\zeta_2\) and \(\zeta_3\) as
\begin{align*}
    \zeta_2 = \; & \tau_i (-1)^{m+1} \left(\sum_{k=1}^{m-1}\varepsilon^{k-1} n_+^{(k)}\right)^{m-1}\partial_Y n_+^{(1)} \\
    & + \tau_i (-1)^{m+1} \left(\sum_{k=1}^{m-1}\varepsilon^{k-1} n_+^{(k)}\right)^{m-1}\partial_Y \left(\sum_{k=2}^{m-1}\varepsilon^{k-1} n_+^{(k)}\right), \\[8pt]
    \zeta_3 = \; & \tau_i \partial_Y  n_{+}^{(m)} + \tau_i \partial_Y \left(\sum_{k=2}^{m} \frac{(-1)^{k+1}}{k} \sum_{j=1}^{k} \binom{k}{j} (N_{+}^{(m-1)} - 1)^{k-j} \varepsilon^{m(j-1)}\left( n_{+}^{(m)} \right)^j \right).
\end{align*}
Utilizing the Multinomial Theorem, \(\zeta_2\) expands to
\begin{align*}
    \zeta_2 = \; & \tau_i (-1)^{m+1} \left(\sum_{\substack{a_1 + a_2 + \cdots + a_{m-1} = m-1 \\ a_1,a_2,\ldots,a_{m-1} \geq 0}} \binom{m-1}{a_1,a_2,\ldots,a_{m-1}} \prod_{i=1}^{m-1} \left(\varepsilon^{i-1}n_+^{(i)}\right)^{a_i} \right)\partial_Y n_+^{(1)} \\
    &  + \tau_i (-1)^{m+1} \left(\sum_{k=1}^{m-1}\varepsilon^{k-1} n_+^{(k)}\right)^{m-1}\partial_Y \left(\sum_{k=2}^{m-1}\varepsilon^{k-1} n_+^{(k)}\right).
\end{align*}
Isolating the \(a_1 = m-1\) term from \(\zeta_2\),
\begin{align*}
    \zeta_2 = \; & \tau_i (-1)^{m+1} \left(\rule{0pt}{24pt}\right. \sum_{\substack{a_1 + a_2 + \cdots + a_{m-1} = m-1 \\ a_1,a_2,\ldots,a_{m-1} \geq 0 \\ a_1 \neq m-1}} \binom{m-1}{a_1,a_2,\ldots,a_{m-1}} \prod_{i=1}^{m-1} \left(\varepsilon^{i-1}n_+^{(i)}\right)^{a_i} \left.\rule{0pt}{24pt}\right)\partial_Y n_+^{(1)} \\
    & + \tau_i (-1)^{m+1} \left(n_+^{(1)}\right)^{m-1} \partial_Y n_+^{(1)} + \tau_i (-1)^{m+1} \left(\sum_{k=1}^{m-1}\varepsilon^{k-1} n_+^{(k)}\right)^{m-1}\partial_Y \left(\sum_{k=2}^{m-1}\varepsilon^{k-1} n_+^{(k)}\right)
\end{align*}

Going back to ion velocity residual, recall 
\begin{equation*}
\text{Res}_{v_+}^{(m)} = \varepsilon^{-1} \text{Res}_{v_+}^{(m-1)} - \varepsilon^{3/2}\Big[\zeta_1 + \zeta_2 + \zeta_3 \Big].
\end{equation*} 
Our goal is to eliminate all terms of magnitude bigger than \(\varepsilon^{5/2}\). By replacement of \(\text{Res}_{v_+}^{(m-1)}\) with the induction hypothesis and our rewritings of \(\zeta_1\), \(\zeta_2\), \(\zeta_3\), is
\begin{align*}
\text{Res}_{v_+}^{(m)} & = -\varepsilon^{5/2}\Big[\partial_T v_{+}^{(m)} + \alpha + \beta_1 + \beta_2 + \beta_3 + \gamma \Big] -\varepsilon^{3/2}\Big[\delta_1 + \delta_2 + \delta_3 \Big], \\[8pt]
    \alpha & \coloneqq \tfrac{1}{2}\partial_Y\left(\sum_{k=2}^{m-1} \sum_{j=k}^{m-1} \varepsilon^{k-2} v_{+}^{(j)} v_{+}^{(m-1+k-j)} \right) + \partial_Y \left( v_{+}^{(m)} \sum_{j=1}^{m-1} \varepsilon^{j-1} v_{+}^{(j)} \right) + \varepsilon^{m-1} v_{+}^{(m)}\partial_Y v_{+}^{(m)}, \\[2pt]
    \beta_1 & \coloneqq \tau_i (-1)^{m+1} \left(\sum_{k=1}^{m-1}\varepsilon^{k-1} n_+^{(k)}\right)^{m-1}\partial_Y \left(\sum_{k=2}^{m-1}\varepsilon^{k-2} n_+^{(k)}\right), \\[2pt]
    \beta_2 & \coloneqq \tau_i (-1)^{m+1} \varepsilon^{-1}\left(\sum_{\substack{S_{m-1} = m-1 \\ a_1,a_2,\ldots,a_{m-1} \geq 0 \\ a_1 \neq m-1}} \binom{m-1}{a_1,a_2,\ldots,a_{m-1}} \prod_{i=1}^{m-1} \left(\varepsilon^{i-1}n_+^{(i)}\right)^{a_i} \right)\partial_Y n_+^{(1)}, \\[2pt]
    \beta_3 & \coloneqq \tau_i \partial_Y\left( \sum_{k=2}^{m-1}\frac{(-1)^{k+1}  \varepsilon^{k-m-1}}{k}\left(\sum_{\substack{S_{m-2} = k \\ P_{m-2} > m-k \\ a_1,a_2,\ldots,a_{m-2} \geq 0 }} \binom{k}{a_1,a_2,\ldots,a_{m-2}} \prod_{i=1}^{m-2} \left(\varepsilon^{i-1}n_+^{(i)} \right)^{a_i}\right) \right), \\[2pt]
    \gamma & \coloneqq \tau_i \partial_Y \left(\sum_{k=2}^{m} \frac{(-1)^{k+1}}{k} \sum_{j=1}^{k} \binom{k}{j} \varepsilon^{m(j-1) + k - j - 1} \left(\sum_{i=1}^{m-1} \varepsilon^{i-1} n_+^{(i)} \right)^{k-j}\left( n_{+}^{(m)} \right)^j \right), \\[2pt]
    \delta_1 & \coloneqq - \tau_i \partial_Y \left( \sum_{k=1}^{m}\frac{(-1)^{k+1}}{k}\left(\sum_{\substack{S_{m} = k \\ a_1 + 2a_2 + \cdots + m a_{m} = m \\ a_1,a_2,\ldots,a_{m} \geq 0 \\ a_1 \neq m \;\&\; a_m \neq 1}} \binom{k}{a_1,a_2,\ldots,a_{m}} \prod_{i=1}^{m} \left(n_+^{(i)}\right)^{a_i}\right)\right),\\[2pt]
    \delta_2 & \coloneqq \tau_i\partial_Y \left(\sum_{k=2}^{m-1} \frac{(-1)^{k+1}}{k} \left(\sum_{j=1}^{k} \binom{k}{j} \varepsilon^{(m-1)(j-1) + k - j - 1} \left(\sum_{i=1}^{m-2} \varepsilon^{i-1} n_+^{(i)} \right)^{k-j} \left( n_{+}^{(m-1)}\right)^{j}\right)\right), \\[2pt]
    \delta_3 & \coloneqq \tau_i \partial_Y\left( \sum_{k=2}^{m-1}\frac{(-1)^{k+1}}{k}  \varepsilon^{k-m}\left(\sum_{\substack{S_{m-2} = k \\ P_{m-2} = m-k \\ a_1,a_2,\ldots,a_{m-2} \geq 0 }} \binom{k}{a_1,a_2,\ldots,a_{m-2}} \prod_{i=1}^{m-2} \left(\varepsilon^{i-1}n_+^{(i)} \right)^{a_i}\right) \right).
\end{align*}
We are continuing to use the notation
\begin{equation*}
S_{val} = \sum_{i=1}^{val} a_i, \quad P_{val} = \sum_{i=1}^{val} (i-1) a_i.
\end{equation*}
It is important to note that \(\beta_2\) is order \(\varepsilon^0\) and higher. The only possible issue of a negative exponent seems like it would appear when \(i=1\) and \(a_1>0\) with \(a_i=0\) for \(2\leq i \leq m-1\), but that only occurs when \(a_1=m-1\), which is already eliminated.

The entirety of \(\delta_1\) can be cancelled out using \(\delta_2\) and \(\delta_3\). To see this, start by reworking \(\delta_3\) as
\begin{equation*}
\delta_3 = \tau_i \partial_Y\left( \sum_{k=2}^{m-1}\frac{(-1)^{k+1}}{k}  \left(\sum_{\substack{S_{m-2} = k \\ P_{m-2} = m-k \\ a_1,a_2,\ldots,a_{m-2} \geq 0 }} \binom{k}{a_1,a_2,\ldots,a_{m-2}} \prod_{i=1}^{m-2} \left(n_+^{(i)} \right)^{a_i}\right) \right),
\end{equation*} 
observing that it has no order of \(\varepsilon\). Next, we apply the Multinomial Theorem to \(\delta_2\),
\begin{align*}
    \delta_2 & = \tau_i\partial_Y \left(\sum_{k=2}^{m-1} \frac{(-1)^{k+1}}{k} \left(\sum_{j=1}^{k} \binom{k}{j} \varepsilon^{(m-1)(j-1) + k - j - 1} \delta_{2A} \left( n_{+}^{(m-1)}\right)^{j} \right)\right), \\[2pt]
    \delta_{2A} & \coloneqq \sum_{\substack{S_{m-2} = k-j \\ a_1,a_2,\ldots,a_{m-2} \geq 0 }}  \binom{k-j}{a_1,a_2,\ldots,a_{m-2}}  \prod_{i=1}^{m-2} \left(\varepsilon^{i-1} n_+^{(i)} \right)^{a_i}.
\end{align*}
We want to isolate the only piece of \(\delta_{2}\) with no \(\varepsilon\). This happens in \(\delta_{2A}\) when \(a_1=1\), \(a_i = 0\) for \(2\leq i \leq m-2\), \(j=1\), and \(k=2\). 
It is 
%
\(\tau_i\partial_Y \left( - n_+^{(1)} n_{+}^{(m-1)} \right)\). Since the remaining terms of \(\delta_2\) all have order \(\varepsilon\) or higher, 
\begin{equation*}
P_{m-2} = \sum_{i=1}^{m-2} a_i(i-1) > -(m-1)(j-1) - k + j + 1.
\end{equation*} 
Ergo,
\begin{align*}
    \delta_{2} =  \tau_i\partial_Y & \left( - n_+^{(1)} n_{+}^{(m-1)}\right) + \tau_i\partial_Y \left(\sum_{k=2}^{m-1} \frac{(-1)^{k+1}}{k} \left(\sum_{j=1}^{k} \binom{k}{j} \varepsilon^{(m-1)(j-1) + k - j - 1} \delta_{2B} \left( n_{+}^{(m-1)}\right)^{j} \right)\right), \\[8pt]
    \delta_{2B} & \coloneqq \sum_{\substack{S_{m-2} = k-j \\ Pa_{m-2} > -(m-1)(j-1) - k + j + 1 \\ a_1,a_2,\ldots,a_{m-2} \geq 0 }}  \binom{k-j}{a_1,a_2,\ldots,a_{m-2}}  \prod_{i=1}^{m-2} \left(\varepsilon^{i-1} n_+^{(i)} \right)^{a_i}.
\end{align*}
Extracting \(\tau_i\partial_Y \left( - n_+^{(1)} n_{+}^{(m-1)} \right)\) and summing it with \(\delta_3\), we get
\begin{align*}
    \tau_i\partial_Y \left( - n_+^{(1)} n_{+}^{(m-1)} \right) + \delta_3 = \; & \tau_i\partial_Y \left( - n_{+}^{(m-1)} n_+^{(1)} \right) + \tau_i \partial_Y\left( \sum_{k=2}^{m-1}\frac{(-1)^{k+1}}{k}\delta_{3A} \right), \\
    \delta_{3A} \coloneqq & \sum_{\substack{S_{m-2} = k \\ P_{m-2} = m-k \\ a_1,a_2,\ldots,a_{m-2} \geq 0 }} \binom{k}{a_1,a_2,\ldots,a_{m-2}} \prod_{i=1}^{m-2} \left(n_+^{(i)} \right)^{a_i}.
\end{align*}
From \(\delta_{3A}\), \(S_{m-2} = k\) and \(P_{m-2} = m-k\) implies
\(
    a_1 + 2a_2 + 3a_3 + \cdots + (m-2)a_{m-2} = m.
\)
It also tells us \(a_1 \neq m\) because \(k \leq m-1\). Furthermore, we can take \(a_{m-1} = 0\) and incorporate it into \(\delta_{3A}\), which markedly excludes the possibility of the coefficient combination of \(a_1 = 1\) with \(a_{m-1} = 1\). So, 
\vspace{-8pt}
\begin{equation*}
\delta_{3A} = \sum_{\substack{S_{m-1} = k \\ a_1 + 2a_2 + \cdots + (m-1) a_{m-1} = m \\ a_1,a_2,\ldots,a_{m-1} \geq 0 \\ a_1 \neq m \;,\; a_1\neq 1 \;\cup\; a_{m-1}\neq 1 }} \binom{k}{a_1,a_2,\ldots,a_{m-1}} \prod_{i=1}^{m-1} \left(n_+^{(i)}\right)^{a_i}.
\end{equation*}
That being said, the term
\begin{equation*}
\tau_i\partial_Y \left( - n_+^{(1)} n_{+}^{(m-1)} \right) = \tau_i \partial_Y \left(\frac{(-1)^{2+1}}{2}\binom{2}{1,0,\ldots,0,1} n_+^{(1)} n_+^{(m-1)} \right),
\end{equation*} 
does accommodate the coefficient combination of \(a_1 = 1\) with \(a_{m-1} = 1\) and has the same form as \(\delta_{3A}\). Therefore,
\begin{align*}
    & \tau_i\partial_Y \left( - n_+^{(1)} n_{+}^{(m-1)} \right) + \delta_3 \\
    & \; = \tau_i \partial_Y \left(\rule{0pt}{24pt}\right. \sum_{k=2}^{m-1}\frac{(-1)^{k+1}}{k}\left(\rule{0pt}{24pt}\right.\sum_{\substack{S_{m-1} = k \\ a_1 + 2a_2 + \cdots + (m-1) a_{m-1} = m \\ a_1,a_2,\ldots,a_{m-1} \geq 0 \\ a_1 \neq m}} \binom{k}{a_1,a_2,\ldots,a_{m-1}} \prod_{i=1}^{m-1} \left(n_+^{(i)}\right)^{a_i}\left.\rule{0pt}{24pt}\right) \left.\rule{0pt}{24pt}\right).
\end{align*}
We will further this equation by inputting \(a_m = 0\), i.e. \(a_m \neq 1\). This is
\begin{align*}
    & \tau_i\partial_Y \left( - n_+^{(1)} n_{+}^{(m-1)} \right) + \delta_3 \\
    &\;  = \tau_i \partial_Y \left(\rule{0pt}{24pt}\right.  \sum_{k=1}^{m}\frac{(-1)^{k+1}}{k}\left(\rule{0pt}{24pt}\right. \sum_{\substack{S_{m} = k \\ a_1 + 2a_2 + \cdots + m a_{m} = m \\ a_1,a_2,\ldots,a_{m} \geq 0 \\ a_1 \neq m \;\&\; a_m \neq 1}} \binom{k}{a_1,a_2,\ldots,a_{m}} \prod_{i=1}^{m} \left(n_+^{(i)}\right)^{a_i}\left.\rule{0pt}{24pt}\right)  \left.\rule{0pt}{24pt}\right) \\
    & \; = -\delta_1.
\end{align*}
Finally, we have removed all terms of magnitude greater than \(\varepsilon^{3/2}\).

What remains of the ion velocity error-scaled residual now is 
\begin{align*}
    \text{Res}_{v_+}^{(m)} = -\varepsilon^{5/2}\Big[\partial_T v_{+}^{(m)} + \alpha + \beta_1 + \beta_2 + \beta_3 + \beta_4 + \gamma \Big]
\end{align*}
where \(\alpha\), \(\beta_1\), \(\beta_2\), \(\beta_3\), and \(\gamma\) are defined as before, and
\begin{align*}
    \beta_4 & = \tau_i\partial_Y \left(\sum_{k=2}^{m-1} \frac{(-1)^{k+1}}{k} \left(\sum_{j=1}^{k} \binom{k}{j} \varepsilon^{(m-1)(j-1) + k - j - 2} \delta_{2B} \left( n_{+}^{(m-1)}\right)^{j} \right)\right),\\[2pt]
    \delta_{2B} & = \sum_{\substack{S_{m-2} = k-j \\ P_{m-2} > -(m-1)(j-1) - k + j + 1 \\ a_1,a_2,\ldots,a_{m-2} \geq 0 }}  \binom{k-j}{a_1,a_2,\ldots,a_{m-2}}  \prod_{i=1}^{m-2} \left(\varepsilon^{i-1} n_+^{(i)} \right)^{a_i}.
\end{align*}
We will continue to rewrite the residual until it is in the form of the induction hypothesis. Beginning with \(\alpha\),
\begin{align*}
    \alpha & = \tfrac{1}{2}\partial_Y\left(\sum_{k=2}^{m-1} \sum_{j=k}^{m-1} \varepsilon^{k-2} v_{+}^{(j)} v_{+}^{(m-1+k-j)} \right) + \partial_Y \left( v_{+}^{(m)} \sum_{j=1}^{m-1} \varepsilon^{j-1} v_{+}^{(j)} \right) + \varepsilon^{m-1} v_{+}^{(m)}\partial_Y v_{+}^{(m)} \\
    & = \tfrac{1}{2}\partial_Y\left(\sum_{k=2}^{m-1} \sum_{j=k}^{m-1} \varepsilon^{k-2} v_{+}^{(j)} v_{+}^{(m-1+k-j)} + \sum_{j=1}^{m-1} \varepsilon^{j-1} v_{+}^{(j)} v_{+}^{(m)} + \sum_{j=1}^{m-1} \varepsilon^{j-1} v_{+}^{(m)} v_{+}^{(j)} + \varepsilon^{m-1} v_{+}^{(m)}v_{+}^{(m)} \right) \\
    & = \tfrac{1}{2}\partial_Y\Bigg( \sum_{k=1}^{m-1} \sum_{j=k+1}^{m-1} \varepsilon^{k-1} v_{+}^{(j)} v_{+}^{(m+k-j)} + \sum_{k=1}^{m-1} \varepsilon^{k-1} v_{+}^{(k)} v_{+}^{(m+k-k)} \\
    & \phantom{\tfrac{1}{2}\partial_Y\Bigg(}\quad + \sum_{k=1}^{m-1} \varepsilon^{k-1} v_{+}^{(m)} v_{+}^{(m + k - m)} + \varepsilon^{m-1} v_{+}^{(m)} v_{+}^{(m)} \Bigg) \\
    & = \tfrac{1}{2}\partial_Y\left( \sum_{k=1}^{m-1} \sum_{j=k}^{m} \varepsilon^{k-1} v_{+}^{(j)} v_{+}^{(m+k-j)} + \varepsilon^{m-1} v_{+}^{(m)} v_{+}^{(m+m-m)} \right)\\[2pt]
    & = \tfrac{1}{2}\partial_Y\left(\sum_{k=1}^{m} \sum_{j=k}^{m} \varepsilon^{k-1} v_{+}^{(j)} v_{+}^{(m+k-j)} \right).
\end{align*}

Next, we will look at the sum \(\beta_1 + \beta_2 + \beta_3 + \beta_4\), starting with \(\beta_1\) and \(\beta_2\). Notice that
\begin{align*}
    \beta_1 + \beta_2 = \; & \tau_i (-1)^{m+1} \left(\sum_{k=1}^{m-1}\varepsilon^{k-1} n_+^{(k)}\right)^{m-1}\partial_Y \left(\sum_{k=2}^{m-1}\varepsilon^{k-2} n_+^{(k)}\right) \\
    & + \tau_i (-1)^{m+1} \varepsilon^{-1}\left(\rule{0pt}{24pt}\right. \sum_{\substack{S_{m-1} = m-1 \\ a_1,a_2,\ldots,a_{m-1} \geq 0 \\ a_1 \neq m-1}} \binom{m-1}{a_1,a_2,\ldots,a_{m-1}} \prod_{i=1}^{m-1} \left(\varepsilon^{i-1}n_+^{(i)}\right)^{a_i} \left.\rule{0pt}{24pt}\right) \partial_Y n_+^{(1)}. 
\end{align*}
Undoing the Multinomial Theorem, their sum becomes
\begin{align*}
    \beta_1 + \beta_2 = \; & \tau_i (-1)^{m+1}\varepsilon^{-1}\left(\sum_{k=1}^{m-1} \varepsilon^{k-1}n_+^{(k)} \right)^{m-1}\partial_Y\left(\sum_{k=2}^{m-1} \varepsilon^{k-1}  n_+^{(k)} \right) \\
    & + \tau_i (-1)^{m+1}\varepsilon^{-1}\left(\sum_{k=1}^{m-1} \varepsilon^{k-1}n_+^{(k)} \right)^{m-1}\partial_Y n_+^{(1)} - \tau_i \partial_Y\left(\frac{(-1)^{m+1}\varepsilon^{-1}}{m} \left(n_+^{(1)}\right)^m \right).
\end{align*}
Combining series and derivatives yields
\begin{align*}
    \beta_1 + \beta_2 & = \tau_i (-1)^{m+1}\varepsilon^{-1}\left(\sum_{k=1}^{m-1} \varepsilon^{k-1}n_+^{(k)} \right)^{m-1}\partial_Y\left(\sum_{k=1}^{m-1} \varepsilon^{k-1}  n_+^{(k)} \right) \\[2pt]
    & \phantom{=}\; - \tau_i \partial_Y\left(\frac{(-1)^{m+1}\varepsilon^{-1}}{m} \left(n_+^{(1)}\right)^m \right) \\[2pt]
    & = \tau_i \partial_Y\left(\frac{(-1)^{m+1}\varepsilon^{-1}}{m}\left(\left(\sum_{k=1}^{m-1} \varepsilon^{k-1}n_+^{(k)} \right)^m - \left(n_+^{(1)}\right)^m \right)\right).
\end{align*}
Reapplying the Multinomial Theorem, this is the same as
\begin{align*}
     \beta_1 + \beta_2 & = \tau_i \partial_Y\left(\rule{0pt}{40pt}\right.\frac{(-1)^{m+1}\varepsilon^{m-m-1}}{m} \left(\rule{0pt}{38pt}\right. \sum_{\substack{S_{m-1} = m \\ P_{m-1} \geq 1 \\ a_1 \neq m \\ a_1,a_2,\ldots,a_{m-1} \geq 0 }} \binom{m}{a_1,a_2,\ldots,a_{m-1}} \prod_{i=1}^{m-1} \left(\varepsilon^{i-1}n_+^{(i)} \right)^{a_i}\left.\rule{0pt}{38pt}\right)  \left.\rule{0pt}{40pt}\right). 
\end{align*}
Now, observe we can take \(a_{m-1}=0\) in \(\beta_3\) so that
\begin{align*}
    \beta_3 & = \tau_i \partial_Y \left(\rule{0pt}{40pt}\right.  \sum_{k=2}^{m-1}\frac{(-1)^{k+1}  \varepsilon^{k-m-1}}{k} \left(\rule{0pt}{38pt}\right. \sum_{\substack{S_{m-1} = k \\ P_{m-1} > m-k \\ a_1,a_2,\ldots,a_{m-2} \geq 0 \\ a_{m-1} = 0}} \binom{k}{a_1,a_2,\ldots,a_{m-1}} \prod_{i=1}^{m-1} \left(\varepsilon^{i-1}n_+^{(i)} \right)^{a_i} \left.\rule{0pt}{38pt}\right)  \left.\rule{0pt}{40pt}\right) .
\end{align*}
Moving along, we will look at \(\beta_4\),
\begin{align*}
    \beta_4 & = \tau_i\partial_Y \left(\sum_{k=2}^{m-1} \frac{(-1)^{k+1}}{k} \left(\sum_{j=1}^{k} \binom{k}{j} \varepsilon^{(m-1)(j-1) + k - j - 2} \delta_{2B} \left( n_{+}^{(m-1)}\right)^{j} \right)\right), \\
    \delta_{2B} & = \sum_{\substack{S_{m-2} = k-j \\ P_{m-2} > -(m-1)(j-1) - k + j + 1 \\ a_1,a_2,\ldots,a_{m-2} \geq 0 }}  \binom{k-j}{a_1,a_2,\ldots,a_{m-2}}  \prod_{i=1}^{m-2} \left(\varepsilon^{i-1} n_+^{(i)} \right)^{a_i}.
\end{align*}
For \(\beta_4\), consider \(a_{m-1}=j\) for some \(1 \leq j \leq k \). We will establish the values of various quantities in \(\beta_4\). First,
\begin{align*}
    \binom{k}{j}\binom{k-j}{a_1,a_2,\ldots,a_{m-2}} & = \frac{k!}{j!(k-j)!}\binom{k-j}{a_1,a_2,\ldots,a_{m-2}} \\[8pt]
    & = \binom{k}{a_1,a_2,\ldots,a_{m-2},j} \\[8pt]
    & = \binom{k}{a_1,a_2,\ldots,a_{m-2},a_{m-1}}.
\end{align*}
For the power of \(\varepsilon\),
\begin{align*}
    (m-1)(j-1) + k - j - 2 & = k -(m-1)-2+j(m-1) + j - 2j \\
    & = k-m-1+j(m-2) \\
    & = k-m-1 + a_{m-1}(m-2).  
\end{align*}
This also makes 
\begin{align*}
    a_2 + 2a_3 + \cdots + (m-3)a_{m-2} & > -(m-1)(j-1) - k + j + 1 \\
    \iff a_2 + 2a_3 + \cdots + (m-3)a_{m-2} + (m-2)a_{m-1} & > -(m-1)(j-1) - k + j + 1 + (m-2)j \\
    \iff a_2 + 2a_3 + \cdots + (m-3)a_{m-2} + (m-2)a_{m-1} & > m - k \\
    \iff a_2 + 2a_3 + \cdots + (m-3)a_{m-2} + (m-2)a_{m-1} & \geq m+1 - k,
\end{align*}
and
\begin{align*}
    a_1 + \cdots + a_{m-2} & = k-j \\
    \iff a_1 + \cdots + a_{m-2} + a_{m-1} & = k-j + j \\
    \iff a_1 + \cdots + a_{m-2} + a_{m-1} & = k.
\end{align*}
Thus, letting 
\begin{equation*}
    \beta_{4A} = (\varepsilon^{m-2}n_+^{(m-1)})^{a_{m-1}}\prod_{i=1}^{m-2} \left(\varepsilon^{i-1}n_+^{(i)} \right)^{a_i},
\end{equation*}
we have
\begin{align*}
    \beta_4 & = \tau_i\partial_Y\left(\rule{0pt}{40pt}\right.\sum_{k=2}^{m-1}\frac{(-1)^{k+1}  \varepsilon^{k-m-1}}{k}\left(\rule{0pt}{38pt}\right.\sum_{\substack{S_{m-1} = k \\ P_{m-1} \geq m+1-k \\ a_1,a_2,\ldots,a_{m-1} \geq 0 \\ a_{m-1} \neq 0}} \binom{k}{a_1,a_2,\ldots,a_{m-1}} \beta_{4A} \left.\rule{0pt}{38pt}\right) \left.\rule{0pt}{40pt}\right).
\end{align*}

Summing \(\beta_1\), \(\beta_2\), \(\beta_3\), \(\beta_4\) up,
\begin{align*}
    & \beta_1 + \beta_2 + \beta_3 + \beta_4 \\
    & = \tau_i \partial_Y \left(\rule{0pt}{40pt}\right. \frac{(-1)^{m+1}\varepsilon^{m-m-1}}{m} \left(\rule{0pt}{38pt}\right. \sum_{\substack{S_{m-1} = m \\ P_{m-1} \geq 1 \\ a_1 \neq m \\ a_1,a_2,\ldots,a_{m-1} \geq 0 }} \binom{m}{a_1,a_2,\ldots,a_{m-1}} \prod_{i=1}^{m-1} \left(\varepsilon^{i-1}n_+^{(i)} \right)^{a_i}\left.\rule{0pt}{38pt}\right) \left.\rule{0pt}{40pt}\right) \\[4pt]
     & \quad + \tau_i \partial_Y\left(\rule{0pt}{40pt}\right. \sum_{k=2}^{m-1}\frac{(-1)^{k+1}  \varepsilon^{k-m-1}}{k}\left(\rule{0pt}{38pt}\right. \sum_{\substack{S_{m-1} = k \\ P_{m-1} > m-k \\ a_1,a_2,\ldots,a_{m-2} \geq 0 \\ a_{m-1} = 0}} \binom{k}{a_1,a_2,\ldots,a_{m-1}} \prod_{i=1}^{m-1} \left(\varepsilon^{i-1}n_+^{(i)} \right)^{a_i} \left.\rule{0pt}{38pt}\right) \left.\rule{0pt}{40pt}\right) \\[4pt]
     & \quad + \tau_i\partial_Y\left(\rule{0pt}{40pt}\right. \sum_{k=2}^{m-1}\frac{(-1)^{k+1}  \varepsilon^{k-m-1}}{k}\left(\rule{0pt}{38pt}\right.\sum_{\substack{S_{m-1} = k \\ P_{m-1} \geq m+1-k \\ a_1,a_2,\ldots,a_{m-1} \geq 0 \\ a_{m-1} \neq 0}} \binom{k}{a_1,a_2,\ldots,a_{m-1}} \prod_{i=1}^{m-1} \left(\varepsilon^{i-1}n_+^{(i)} \right)^{a_i} \left.\rule{0pt}{38pt}\right) \left.\rule{0pt}{40pt}\right) \\[8pt]
    & = \tau_i \partial_Y \left(\rule{0pt}{40pt}\right. \sum_{k=2}^{m}\frac{(-1)^{k+1}  \varepsilon^{k-m-1}}{k}\left(\rule{0pt}{38pt}\right.\sum_{\substack{S_{m-1} = k \\ P_{m-1} \geq m+1-k \\ a_1,a_2,\ldots,a_{m-1} \geq 0 }} \binom{k}{a_1,a_2,\ldots,a_{m-1}} \prod_{i=1}^{m-1} \left(\varepsilon^{i-1}n_+^{(i)} \right)^{a_i}\left.\rule{0pt}{38pt}\right) \left.\rule{0pt}{40pt}\right) .
\end{align*}
Therefore,
    \begin{align*}
        & \text{Res}_{v_+}^{(m)} = -\varepsilon^{5/2}\left[\rule{0pt}{28pt}\right. \partial_T v_+^{(m)} + \tfrac{1}{2}\partial_Y\left(\sum_{k=1}^{m} \sum_{j=k}^{m} \varepsilon^{k-1} v_{+}^{(j)} v_{+}^{(m+k-j)} \right) \\[10pt] 
        & + \tau_i \partial_Y\left( \sum_{k=2}^{m}\frac{(-1)^{k+1}\varepsilon^{k-m-1}}{k}  \left(\sum_{\substack{S_{m-1} = k \\ P_{m-1} \geq m+1-k}} \binom{k}{a_1,a_2,\ldots,a_{m-1}} \prod_{i=1}^{m-1} \left(\varepsilon^{i-1}n_+^{(i)} \right)^{a_i}\right) \right) \\[10pt] 
        & + \tau_i\partial_Y \left(\sum_{k=2}^{m} \frac{(-1)^{k+1}}{k} \left(\sum_{j=1}^{k} \binom{k}{j} \varepsilon^{m(j-1) + k - j - 1} \left(\sum_{i=1}^{m-1} \varepsilon^{i-1} n_+^{(i)} \right)^{k-j} \left( n_{+}^{(m)}\right)^{j}\right)\right) \left.\rule{0pt}{28pt}\right],
    \end{align*}
    where \(0 \leq a_i \in\mathbb{Z}\) for all \(1\leq i \leq m-1\) and
    \begin{equation*}
    S_{m-1} = \sum_{i=1}^{m-1} a_i, \quad P_{m-1} = \sum_{i=1}^{m-1} (i-1) a_i.
    \end{equation*}
\end{proof}

\begin{lemma}\label{lem::res_vel_elec}
    We can rewrite the electron velocity error-scaled residual terms as 
    \begin{align*}
        & \text{Res}_{v_-}^{(m)} = -\varepsilon^{5/2}\left[\rule{0pt}{28pt}\right. m_e \left(\partial_T v_-^{(m)} + \tfrac{1}{2}\partial_Y\left(\sum_{k=1}^{m} \sum_{j=k}^{m} \varepsilon^{k-1} v_{-}^{(j)} v_{-}^{(m+k-j)} \right)\right) \\[8pt] 
        & + \partial_Y\left( \sum_{k=2}^{m}\frac{(-1)^{k+1}  \varepsilon^{k-m-1}}{k}\left(\sum_{\substack{S_{m-1} = k \\ P_{m-1} \geq m+1-k}} \binom{k}{a_1,a_2,\ldots,a_{m-1}} \prod_{i=1}^{m-1} (\varepsilon^{i-1}n_-^{(i)})^{a_i}\right) \right) \\[8pt]
        & + \partial_Y \left(\sum_{k=2}^{m} \frac{(-1)^{k+1}}{k} \left(\sum_{j=1}^{k} \binom{k}{j} \varepsilon^{m(j-1) + k - j - 1} \left(\sum_{i=1}^{m-1} \varepsilon^{i-1} n_-^{(i)} \right)^{k-j} \left( n_{-}^{(m)}\right)^{j}\right)\right) \left.\rule{0pt}{28pt}\right].
    \end{align*}
where \(0 \leq a_i \in\mathbb{Z}\) for all \(1\leq i \leq m-1\) and 
\begin{equation*}
S_{m-1} = \sum_{i=1}^{m-1} a_i, \quad P_{m-1} = \sum_{i=1}^{m-1} (i-1) a_i.
\end{equation*}
\end{lemma}
\begin{proof}
    The proof of this lemma is nearly identical to the process from Lemma~\eqref{lem::res_vel_ion} and is therefore omitted. The only marked difference comes from the key coefficients \(m_e\) and \(\tau_i\). 
\end{proof}

\begin{lemma}\label{lem::res_phi}
    We can rewrite the electric potential error-scaled residual terms as 
    \begin{align*}
    \text{Res}_{\phi}^{(m)} = -\varepsilon^2\partial_Y^2 \phi^{(m)}.
    \end{align*}
\end{lemma}

\begin{proof}
For our base case, with the application of the Gardner-Morikawa transformation, \begin{align*}\text{Res}_{\phi_{\phantom{+}}}^{(1)} = - \Big[\partial_x^2 \Phi^{(1)} - N_{-}^{(1)} + N_{+}^{(1)}\Big] = - \Big[\varepsilon \partial_Y^2 \Phi^{(1)} - N_{-}^{(1)} + N_{+}^{(1)}\Big].\end{align*}
Expanding the terms and utilizing the fact that \(n_{-}^{(1)} = n_{+}^{(1)}\) from \eqref{eq::rely_on_n+1},
\begin{align*}
    \text{Res}_{\phi_{\phantom{+}}}^{(1)} & = -\Big[\varepsilon^2 \partial_Y^2 \phi^{(1)} - \varepsilon n_{-}^{(1)} + \varepsilon n_{+}^{(1)}\Big] \\
    & = -\varepsilon^2 \Big[\partial_Y^2 \phi^{(1)}\Big].
\end{align*}
    Thus we make our induction hypothesis, \begin{align*}\text{Res}_{\phi}^{(m-1)} = -\varepsilon^2\partial_Y^2 \phi^{(m-1)}.\end{align*}
    For the \(m^{\text{th}}\) step of the electric potential residual, we follow a similar process as before to get
    \begin{align*}
        \text{Res}_{\phi_{\phantom{+}}}^{(m)} & = -\varepsilon^{1-m}\Big[\partial_x^2 \Phi^{(m)} - N_{-}^{(m)} + N_{+}^{(m)}\Big] \\
        & = -\varepsilon^{1-m}\Big[\varepsilon \partial_Y^2 \Phi^{(m)} - N_{-}^{(m)} + N_{+}^{(m)}\Big] \\
        & = -\varepsilon^{1-m}\Big[\varepsilon \partial_Y^2 \Phi^{(m-1)} - N_{-}^{(m-1)} + N_{+}^{(m-1)}\Big] - \varepsilon \Big[\varepsilon\partial_Y^2 \phi^{(m)} - n_{-}^{(m)} + n_{+}^{(m)}\Big] \\
        & = \varepsilon^{-1}\text{Res}_{\phi}^{(m-1)} - \varepsilon \Big[\varepsilon\partial_Y^2 \phi^{(m)} - (n_{+}^{(m)} + \partial_{Y}^{2}\phi^{(m-1)}) + n_{+}^{(m)}\Big] \\
        & = -\varepsilon \partial_Y^2 \phi^{(m-1)} - \varepsilon \Big[\varepsilon\partial_Y^2 \phi^{(m)} - \partial_{Y}^{2}\phi^{(m-1)}\Big] \\
        & = - \varepsilon^2 \partial_Y^2 \phi^{(m)}.
    \end{align*}
\end{proof}


With these revised residuals, we can discover how well our initial approximation solves the EP System \eqref{eq::5hotrescale}. 
Specifics are encapsulated in Theorem~\ref{thm::res_size}.

\begin{center}{\rule{4cm}{0.4pt}}\end{center}

\begin{theorem}\label{thm::res_size}
Fix \(m\geq 1\). Then
\begin{align*}
    \sup_{0 \leq t}\; \lVert \text{Res}_{n_\pm}^{(m)}(\cdot, t) \rVert_{L^2(x)} & \leq C \varepsilon^{9/4} \\
    \sup_{0 \leq t}\; \lVert \text{Res}_{v_\pm}^{(m)}(\cdot, t) \rVert_{L^2(x)} & \leq C \varepsilon^{9/4} \\
    %
    %
    \sup_{0 \leq t}\; \lVert \text{Res}_{\phi}^{(m)}(\cdot, t) \rVert_{L^2(x)} & \leq C \varepsilon^{7/4}
\end{align*}
where \(C\) is a scalar dependent on 
\[
\sup_{0 \leq T}\; \left\lVert \left(N_{\pm}^{(m)}, V_{\pm}^{(m)}, \Phi^{(m)}\right)(\cdot,T)  \right\rVert_{H^{s_m}(Y)}.
\]
\end{theorem}

\begin{proof}
    By Theorem~\ref{thm::5_solve_n1KdV} and \ref{thm::5_solve_nkKdV} for \(2\leq k \leq m\), \(\big(N_{\pm}^{(m)}, V_{\pm}^{(m)}, \Phi^{(m)} \big)~\in L^\infty\left(H^{s_m}(\mathbb{R}) ; -T, T\right) \) for any time interval. The rest of the proof follows from Lemmas \eqref{lem::res_dens}, \eqref{lem::res_vel_ion}, \eqref{lem::res_vel_elec}, and \eqref{lem::res_phi}. For instance,
    \begin{equation*}
        \lVert \text{Res}_{\phi}^{(m)} \rVert_{L^2(x)}^2 = \int \left\lvert \varepsilon^2 \partial_Y^2 \phi^{(m)}\left(\varepsilon^{1/2}(x-ct), \; \varepsilon^{3/2}t)\right) \right\rvert^2 \; dx = \varepsilon^{7/2} \lVert \partial_Y^2 \phi^{(m)}(Y,T) \rVert_{L^2(Y)} \leq C\varepsilon^{7/2}.
    \end{equation*}
\end{proof}

\begin{center}{\rule{4cm}{0.4pt}}\end{center}

%

%% file: conclusion.tex

\section{Conclusion and Further Notes}\label{kdv_1}

We have bounded the residual terms by a small order of \(\varepsilon\). Therefore, we know the initial expansion \(\big(N_{\pm}^{(m)}, V_{\pm}^{(m)}, \Phi^{(m)}\big)\) accurately captures the behavior of \eqref{eq::5hotrescale}. Consequently, the Euler-Poisson system for a two-fluid plasma can be approximated by solutions of the Korteweg-de Vries equations in the long-wavelength limit.

Guo and Pu have established the Korteweg-de Vries limit of the single-fluid system \cite{guo}. However, they went further and obtained a global uniform estimate for the remainder terms \((n_+^{(R_\varepsilon)}, v_+^{(R_\varepsilon)}, \phi^{(R_\varepsilon)})\) in \(H^2\) bounded by the system's initial condition. They utilized pseudo-differential operator techniques and established a new norm to do so. Eventually, we will be able to do something similar. For now, we state the following analogous conjecture.
\begin{conjecture}
Let \((n_\pm^{(k)},v_\pm^{(k))},\phi^{(k)})(\varepsilon^{1/2}(x-ct),\varepsilon^{3/2}t)\) be solutions to \eqref{eq::5_rely_on_n+k} and \eqref{eq::5_KdV_n+k} with initial values \((n_\pm^{(k)},v_\pm^{(k))},\phi^{(k)})(\varepsilon^{1/2}x,0)\in H^{s_k}\) for \(1 \leq k \leq m\) established in Theorems~\ref{thm::5_solve_n1KdV} and \ref{thm::5_solve_nkKdV}. Let \((n_\pm^{(R_\varepsilon)},v_\pm^{(R_\varepsilon)},\phi^{(R_\varepsilon)})(x,0) \in H^{s}\), \(s\geq 3\), satisfy \eqref{eq::5_remainder_sys} and assume
\begin{align*}
    n_\pm(x,0) & = N_{\pm}^{(m)}\left(\varepsilon^{1/2}x,0\right) + \varepsilon^{m-1} n_{\pm}^{(R_\varepsilon)}(x,0), \\
    v_\pm(x,0) & = V_{\pm}^{(m)}\left(\varepsilon^{1/2}x ,0\right) + \varepsilon^{m-1} v_{\pm}^{(R_\varepsilon)}(x,0), \\
    \phi(x,0) & = \Phi^{(m)}\left(\varepsilon^{1/2}x,0\right) + \varepsilon^{m-1}\phi^{(R_\varepsilon)}(x,0).
\end{align*}
Then for any \(t>0\), there exists \(\varepsilon^*>0\) such that if \(0 < \varepsilon < \varepsilon^*\), the solution of the two-fluid EP system \eqref{eq::5hotrescale} with initial data \((n_\pm, v_\pm, \phi)(x,0)\) can be expressed as
\begin{align*}
    n_\pm(x,t) & = N_{\pm}^{(m)}\left(\varepsilon^{1/2}(x-ct),\varepsilon^{3/2}t\right) + \varepsilon^{m-1} n_{\pm}^{(R_\varepsilon)}(x,t), \\
    v_\pm(x,t) & = V_{\pm}^{(m)}\left(\varepsilon^{1/2}(x-ct),\varepsilon^{3/2}t\right) + \varepsilon^{m-1} v_{\pm}^{(R_\varepsilon)}(x,t), \\
    \phi(x,t) & = \Phi^{(m)}\left(\varepsilon^{1/2}(x-ct),\varepsilon^{3/2}t\right) + \varepsilon^{m-1} \phi^{(R_\varepsilon)}(x,t).
\end{align*}
Moreover, for all \(0 < \varepsilon < \varepsilon^*\),
    \begin{align*}
    \sup_{0 \leq t}\; \left\lVert (n_\pm^{(R_\varepsilon)},v_\pm^{(R_\varepsilon)},\phi^{(R_\varepsilon)})(\cdot, t) \right\rVert_{H^{\delta_1}}^{2} & \leq C_t\left(1 + \lVert (n_\pm^{(R_\varepsilon)},v_\pm^{(R_\varepsilon)},\phi^{(R_\varepsilon)})(\cdot,0) \rVert_{H^{\delta_2}}^2\right)
    \end{align*}
where \(\lVert \cdot \rVert_{H^{\delta_1}}\) and \(\lVert \cdot \rVert_{H^{\delta_2}}\) are to-be-defined norms that combine terms in \(H^{2}\).
\end{conjecture}

%% file: appendix.tex
\begin{appendix}\label{appendix}

\appendix


\section{\large{Ell Derivation}}\label{app::ell_deriv}

We are looking for coefficients of the \(\varepsilon^{k+1}\) term of our formal expansion in Section~\ref{sec::eps_k+1_coeff}. Located in the velocity equations is \(\ln(n_\pm)\), which, with the expansion, is 
\begin{equation*}
    \ln (1+\varepsilon n_\pm^{(1)} + \varepsilon^2 n_\pm^{(2)} + \cdots).
\end{equation*}
In order to obtain the \(\varepsilon^{k+1}\) coefficient, we need to use Taylor's Theorem to approximate the function. So,
\begin{equation*}
    \ln\left(1 + \varepsilon n_\pm^{(1)} + \varepsilon^2 n_\pm^{(2)} + \cdots + \varepsilon^{k+1} n_\pm^{(k+1)}\right) = \sum_{j=1}^{k+1}\frac{(-1)^{j+1}}{j}\left( \sum_{i=1}^{k+1} \varepsilon^i n_\pm^{(i)}\right)^j.
\end{equation*}
We then expand further using Multinomial Theorem, to get 
\begin{align*}
    & = \sum_{j=1}^{k+1}\frac{(-1)^{j+1}}{j}\left(\sum_{\substack{ a_1 + a_2 + \cdots + a_{k+1} = j \\ a_1,a_2,\ldots,a_{k+1} \geq 0}} \binom{j}{a_1,a_2,\ldots,a_{k+1}} \prod_{i=1}^{k+1} \left(\varepsilon^i n_\pm^{(i)}\right)^{a_i}\right) \\[5pt]
    & = \sum_{j=1}^{k+1}\frac{(-1)^{j+1}}{j}\left(\sum_{\substack{ a_1 + a_2 + \cdots + a_{k+1} = j \\ a_1,a_2,\ldots,a_{k+1} \geq 0}} \binom{j}{a_1,a_2,\ldots,a_{k+1}} \prod_{i=1}^{k+1} \varepsilon^{ia_i}\left( n_\pm^{(i)}\right)^{a_i}\right).
\end{align*}
For coefficient of \(\varepsilon^{k+1}\), we need 
\begin{equation*}
    k+1 = \sum_{i=1}^{k+1} i a_i = a_1 + 2a_2 + \cdots + (k+1)a_{k+1}.
\end{equation*} 
Hence, we define
\begin{equation*}
    \ell^{(k+1)}_{\pm} = \sum_{j=1}^{k+1}\frac{(-1)^{j+1}}{j}\left(\sum_{\substack{a_1 + 2a_2 + \cdots + (k+1) a_{k+1} = k+1 \\ a_1 + a_2 + \cdots + a_{k+1} = j \\ a_1,a_2,\ldots,a_{k+1} \geq 0}} \binom{j}{a_1,a_2,\ldots,a_{k+1}} \prod_{i=1}^{k+1} \left(n_\pm^{(i)}\right)^{a_i}\right).
\end{equation*}


\section{\large{KdV solution}}\label{app::kdv_sol}
We are looking for a solution to \eqref{eq::KdV_n+1},
\begin{equation*}
    \partial_{T}n_{+}^{(1)} + c n_{+}^{(1)}\partial_{Y}n_{+}^{(1)} + \frac{(c^2 m_e -1)^2}{2c(1+m_e)} \partial_{Y}^{3}n_{+}^{(1)} = 0.
\end{equation*}
Starting, we make a traveling wave ansatz of speed \(\tilde\mu\): \(n_+^{(1)}(Y,T) = n_+^{(1)}(Y - \tilde\mu T)\coloneqq n_+^{(1)}(\xi)\). Inputting this ansatz into the KdV equation and then integrating with respect to \(\xi\) yields
\begin{equation*}
    -\tilde\mu n_{+}^{(1)} + \frac{c}{2} (n_{+}^{(1)})^2 + \frac{(c^2 m_e -1)^2}{2c(1+m_e)} (n_{+}^{(1)})'' = 0.
\end{equation*}
The constant of integration will be zero as traveling waves are asymptotically null. Multiply each term by \((n_+^{(1)})'\) and integrate again to get
\begin{equation*}
    \frac{-\tilde\mu}{2}(n_+^{(1)})^2 + \frac{c}{6}(n_+^{(1)})^3 + \frac{(c^2 m_e -1)^2}{4c(1+m_e)} ((n_+^{(1)})')^2 = 0.
\end{equation*}
Rewriting, we find
\begin{equation*}
    \frac{d(n_+^{(1)})}{d\xi} = \pm \frac{2 n_+^{(1)}\sqrt{c(1+m_e)(\frac{\tilde\mu}{2} - \frac{c}{6}n_+^{(1)}) }}{c^2 m_e - 1},
\end{equation*}
or
\begin{equation*}
    \pm \frac{(c^2 m_e - 1)d(n_+^{(1)})}{2 n_+^{(1)}\sqrt{c(1+m_e)(\frac{\tilde\mu}{2} - \frac{c}{6}n_+^{(1)}) }} = d\xi.
\end{equation*}
We will integrate both sides and have
\begin{align*}
    & \mp \frac{\sqrt{2}\left(c^2 m_e-1\right) \tanh ^{-1}\left(\frac{\sqrt{3 \tilde\mu -c n_+^{(1)}}}{\sqrt{3\tilde\mu}}\right)}{\sqrt{c \tilde\mu \left(m_e+1\right)}} \;=\; \xi \\[5pt]
    \implies \qquad & \tanh ^{-1}\left(\frac{\sqrt{3 \tilde\mu -c n_+^{(1)}}}{\sqrt{3\tilde\mu}}\right) \;=\; \mp \xi \sqrt{\frac{c \tilde\mu \left(m_e+1\right)}{2\left(c^2 m_e-1\right)^2}}\\[5pt]
    \implies\qquad & n_+^{(1)} \;=\; \frac{3\tilde\mu}{c}\left(1 - \tanh^2\left(\mp \xi \sqrt{\frac{c \tilde\mu \left(m_e+1\right)}{2\left(c^2 m_e-1\right)^2}}\right)\right).
\end{align*}
We will now use the fact that \(1-\tanh^2\theta = \text{sech}^2\theta\). Hence,
\begin{equation*}
    n_+^{(1)} \;=\; \frac{3\tilde\mu}{c} \text{sech}^2\left(\mp \xi \sqrt{\frac{c \tilde\mu \left(m_e+1\right)}{2\left(c^2 m_e-1\right)^2}}\right) .
\end{equation*}
Moreover, because hyperbolic secant is an even function,
\begin{equation*}
    n_+^{(1)} \;=\; \frac{3\tilde\mu}{c} \text{sech}^2\left(\xi \sqrt{\frac{c \tilde\mu \left(m_e+1\right)}{2\left(c^2 m_e-1\right)^2}}\right) .
\end{equation*}
Replacing \(\xi\) with its definition,
\begin{equation*}
    n_+^{(1)}(Y,T) \;=\; \frac{3\tilde\mu}{c} \text{sech}^2\left((Y - \tilde\mu T) \sqrt{\frac{c \tilde\mu \left(m_e+1\right)}{2\left(c^2 m_e-1\right)^2}}\right) .
\end{equation*}

\end{appendix}